\documentclass{article}
\usepackage{graphicx} 
\usepackage{tikz} 

\usepackage{subcaption} 
\DeclareCaptionLabelFormat{panel}{Panel~#2}
\usepackage{amssymb}
\usepackage{amsmath}
\usepackage{amsthm} 

\newtheorem{proposition}{Proposition}
\newtheorem{lemma}{Lemma}

\usepackage[
backend=biber,
style=apa ,
hyperref=true ,
]{biblatex} 
\usepackage[
  colorlinks=true,
  citecolor=blue,
  linkcolor=black,
  urlcolor=blue
]{hyperref}

\usepackage{xcolor}

\title{Portfolio Choice of Spouses with Heterogeneous Risk Preferences over Shared Household Consumption}
\author{Azar Aliyev}
\date{\today}

\begin{document}

\maketitle

\begin{abstract}
A number of recent empirical papers rely on a collective model to analyze the portfolio choice of spouses, their heterogeneous risk preferences, and intra-household bargaining. I show that although Pareto-Efficient, this approach generates a counterintuitive result with constant Pareto weights: an increase in the risk aversion of a household member can lead to an increase in household risk-taking. Moreover, household wealth affects the preferred risky portfolio share even under CRRA preferences and can make the household insensitive to preferences of a spouse. I offer an alternative approach based on the concept of risk premia and show that it always generates a Pareto-efficient and monotonic household level solution. I derive an intuitive normalization to the Pareto weights in the collective model which generates an approximately equivalent solution. Finally, I discuss multiplicative aggregation of utilities previously used in this literature and show that it is also approximately equivalent to the suggested approach. 
\end{abstract}

\section{Introduction}
The household finance literature used to exclusively assume that the household acts as a single decision-making unit. For example, research interested in the relationship between household characteristics and financial outcomes commonly uses education, age, or risk aversion of the ``household head" as a proxy for household characteristics. This is referred to as \textit{unitary approach} and it only works if the household consists of a single member or if all members of the household have identical objectives. 

This paper is concerned about portfolio choice decisions or, more precisely, risk-taking variation across households. It is common to study heterogeneity in the risky share, defined as the share of the portfolio invested in risky assets. The fundamental result of a unitary framework is that an increase in the risk aversion of an agent always leads to a decrease in the risky share. For example, a canonical formula from \cite{merton1969lifetime}, \(\alpha^\star = \frac{\mu}{A\sigma^2} \), suggests that the optimal risky share, \(\alpha^\star\), is inversely related to the risk aversion, \(A\).

However, a household may consist of more than one member and these members may differ in their preferences for risk and such a unitary model may not be suitable\footnote{Partners may also have heterogeneous beliefs about the stock market or investment horizons, but this is not the focus of this paper.}. Empirical estimates using the Health and Retirement Study (HRS), for example, show that a majority of spouses indeed differ in their risk preferences (\cite{barsky1997preference, kimball2008imputing}). 

There is a number of reasons why one might be interested in using measures of individual risk aversion of both spouses. First, one might attempt to better explain the empirical variation in risk-taking of married partners across time and in the cross-section (see for example, \cite{mazzocco2004saving, yilmazer2015portfolio,thornqvist2014bargaining, Inkmann_Michaelides_Zhang_2026}). In particular, this can help explain changes in portfolio choice after transitions to retirement or death of a member and can motivate time-varying household risk aversion (\cite{addoum2017household}). Second, one might be interested in measuring gender inequality or asymmetry in influence over decision-making under uncertainty (\cite{ke2021wears, gu2024gender, carlsson2013influence}). 

This empirical household finance research often explicitly or implicitly relies on a \textit{collective approach} whereby individual utilities are scaled by some Pareto weights and additively aggregated into a household objective. Pioneered by \cite{chiappori1988rational} and \cite{browning1998efficient} the collective modeling approach is widely used in economics (see \cite{jayachandran2025women} for a recent review of applications in a gender economics context)\footnote{There are some alternative approaches not discussed here. For example, one might model exogenous changes in family structure (\cite{love2010effects}). A less common  approach to intra-household decision making is to rely on an information asymmetry framework (\cite{buchmann2025good}).}. For example, \cite{mazzocco2004saving} uses a collective framework to study the consumption-saving behavior of married couples. A number of empirical papers studying portfolio choice also rely on this collective framework. However, I argue that there is a lack of a rigorous and reliable modeling approach.

The first goal of this paper is to emphasize the drawbacks of application of the classical collective framework to the household portfolio choice problem. The second goal is to offer an alternative model of intra-household portfolio choice which obeys the desired intuitive properties.

I first present a \textit{non-monotonicity} result: an increase in the risk aversion of a member may lead to an increase in the household risky share under a collective framework with constant Pareto weights. To the best of my knowledge, this has not been pointed out by the existing literature. As a preview and in order to provide further motivation, I demonstrate the non-monotonicity result with the following example. A household chooses between a risky alternative \(s\), that provides \$1 worth of a good with probability 0.9 and \$60 worth of the good with probability 0.1, and a safe alternative \(f\) that provides a certain \$5 worth of the good. This gamble is taken from \cite{apesteguia2018monotone} which I link some of the findings to as discussed below. Now, assume the household consists of two partners enjoying the same level of consumption equal to the total level available. Assume both partners have CRRA utility over consumption with preference parameters \(\gamma_1\) and \(\gamma_2\) and the household maximizes the sum of expected individual utilities. I fix the risk aversion of the first member at \(\gamma_1=0.08\), and vary the risk aversion of the second member. At low values of the risk aversion of the second member (say \(\gamma_2 = 0.08\)), the household prefers the risky alternative. At a slightly higher value (say \(\gamma_2 = 0.8\)), the household starts preferring the safe alternative. However, if one increases the risk aversion of the second member further (say \(\gamma_2=8\)), the household will now prefer the risky alternative again\footnote{Under \(\gamma_2=0.08\), we get \(\mathbb{E}[U_h(s)]=\mathbb{E}[U_1(s)]+\mathbb{E}[U_2(s)]\approx11.4 > U_h(f) = U_1(f)+U_2(f)\approx9.6 \). Under \(\gamma_2=0.8\), we get \(\mathbb{E}[U_h(s)]=\mathbb{E}[U_1(s)]+\mathbb{E}[U_2(s)]\approx11.3 <U_h(f)= U_1(f)+U_2(f)\approx11.7 \). Finally, under \(\gamma_2=8\), we get \(\mathbb{E}[U_h(s)]=\mathbb{E}[U_1(s)]+\mathbb{E}[U_2(s)]\approx5.5 > U_h(f) = U_1(f)+U_2(f)\approx4.8 \)}.

I place the counterintuitive result obtained in the above example in a more general context. In particular, I generalize to a continuous-choice and continuous-risk setup and show that the optimal household risky share can increase with the risk aversion of its member under a household objective that additively aggregates individual CARA or CRRA utilities that depend on consumption of a common good. I provide intuition for this result, offer a formal characterization of this pattern, and show that it occurs over realistic values of exogenous model parameters. In addition, I relate this result to the earlier theoretical findings of \cite{apesteguia2018monotone} who study the random utility model (RUM) using a unitary framework. 

I also discuss the drawbacks of application of the classical collective approach related to the interplay between the household risky share and the level of household wealth. In particular, I show that while the individual CRRA solution does not, the household solution that additively aggregates individual utilities is sensitive to the level of wealth. More importantly, I show that irrespective of the Pareto weights, for lower levels of wealth, the household closely mimics the preferences of the more risk averse member, and for higher levels of wealth, the household closely mimics the preferences of the less risk averse member.

Finally, I review the existing work on intra-household portfolio choice in light of these results, reconcile different approaches with each other, and point to the potential confusion in the recent literature. As already mentioned, the existing literature has not highlighted the non-monotonicity of household's risk-taking with respect to the risk aversion of its members. Moreover, I show that such monotonicity is sometimes falsely taken for granted (see, \cite{yilmazer2015portfolio}). I discuss the model used in a recent study by \cite{gu2024gender} who obtain a monotonic solution, yet mistakenly claim that their approach can be micro-founded by a collective framework whereby individual utilities are aggregated additively. 

I propose an alternative approach which is based on the concept of risk premium: the household is assumed to choose a risky share which equates the dollar amounts spouses are willing to pay and willing to accept for a marginal decrease in risk-taking. Hence, I refer to this risky share as the equal marginal risk premia solution. I show that such a household risky share is Pareto efficient and a greater individual risk aversion always leads to a decrease in the household risky share. 

I provide numerical examples to the equal marginal risk premia household risky share and show that it is well approximated by a harmonic average of the individually optimal ones. This implies that the effective household risk aversion is an \textit{arithmetic} average of the individual risk aversion parameters. This is in contrast to the approach in \cite{gu2024gender} which assumes that the household risk aversion is a \textit{harmonic} average of the individual risk aversion parameters. I argue that the model used by \cite{gu2024gender} underestimates the bargaining weight of the less risk-averse spouse and since males are less risk-averse in aggregate, I conclude that authors overestimate the average bargaining power of women.

Finally, I show that alternative approaches to the intra-household decision making exhibit approximate equivalence to the proposed approach. In particular, I demonstrate that scaling Pareto weights by the inverses of the respective marginal utilities of spouses in the collective approach results in a solution approximately equal to the one produced by the equal marginal risk premia approach. In addition, I show that the expected product of utilities criteria used by \cite{addoum2017household} also results in approximately the same household risky share.

This paper contributes to the existing literature on multiple fronts. First, it contributes to the understanding of intra-household dynamics and gender asymmetry in the context of portfolio choice by providing a reliable and rigorous modeling approach. There is a growing body of literature on risk-taking behavior concerning intra-household dynamics and gender inequality. Using U.S. data, \cite{ke2021wears} estimates gender asymmetry when it comes to financial risk-taking using a finance-profession indicator as a proxy for finance sophistication.  \cite{gu2024gender} attempt to estimate asymmetry in decision-making using Australian survey data, which contain estimates of individual risk aversion of both spouses. In the context of a rural Chinese town, \cite{carlsson2013influence} discuss a promising experimental approach to estimating asymmetry, yet authors do not provide or discuss a relevant model. Second, it contributes to a broader understanding of household portfolio choice (see \cite{gomes2021household, campbell2026household} for a general review of the portfolio choice literature and section 4.4 of \cite{gomes2021household} for a partial review of the intra-household finance literature).

The paper is organized as follows. Section \ref{section_Setup} presents the setup of the problem and outlines some preliminaries. Section \ref{section_Utilitarian} discusses portfolio choice under collective approach, derives the main negative results and discusses the relevant empirical literature. Section \ref{section_ERP} presents the equal risk premia solution, discusses its properties and demonstrates the equivalences. Section \ref{section_Conclusion} concludes.

\section{Background}
\label{section_Setup}
Assume the household consists of two individuals \(i \in \{1,2\}\). I will refer to these individuals interchangeably as members and spouses. The household first decides how to invest its initial endowment \(W>0\), then returns are realized, and all resulting wealth is consumed by the members. There is no labor income. Matching in the marriage market is left outside the scope of this framework, and risk preferences of the spouses are taken as given. Spouses are fully committed and must decide on the portfolio allocation together.

There are two assets available in the market: a risk-free asset with a safe return \(r_f\) and a risky asset with excess return \(\tilde{x}\). Let \(\mathbb{E}[\tilde{x}]=\mu>0\) and \(Var(\tilde{x})=\sigma^2\). A potential disagreement in beliefs about the risky return distribution is also outside of the scope of this framework. Portfolio allocation is summarized by the choice of the risky share, \(\alpha\), defined as the share of endowment invested in the risky asset. Hence, household-level consumption is 
        \begin{equation}
            C_h=\alpha W(1+r_f+\tilde{x})+(1-\alpha)W(1+r_f)= W(1 + r_f +\alpha \tilde{x})
        \end{equation}
    
There are two extreme assumptions that can be made about individual consumption. I focus on the approach whereby both members are assumed to consume the same good and individual consumption is equal to the total available level: 
\[C_1=C_2=C_h\] 
I therefore refer to this common consumption good as a public good. This approach is commonly taken by the existing portfolio choice literature studying intra-household dynamics (e.g. see \cite{addoum2017household,yilmazer2015portfolio,thornqvist2014bargaining,gu2024gender}). An alternative approach not discussed in this paper is to assume that each member consumes some private good: \(C_1+C_2=C_h\) (e.g. see \cite{mazzocco2004saving, Inkmann_Michaelides_Zhang_2026}) 

I consider the two most common specifications of the individual utility functions:  CARA utility with absolute risk aversion \(A_i>0\): 
\begin{equation}
    U_i(C_i) = \frac{-exp(-A_iC_i)}{A_i}
\end{equation} 
and CRRA utility with relative risk aversion \(\gamma_i>0\): 
\begin{equation}
    U_i(C_i) = \frac{C_i^{1-\gamma_i}}{1-\gamma_i}
\end{equation}

Before proceeding with the intra-household approach, it is useful to consider the traditional individual maximization problem, also referred to as the unitary approach. This yields an individually optimal risky share \(\alpha^\star_i\) which can be thought of as a counterfactual risky share that would have been picked by spouse \(i\) if she had total control over decision-making and maximized her own utility.
    
\begin{equation}
    \alpha^*_i \equiv \operatorname*{arg\,max}_{\alpha} \mathbb{E} [U_{i}(C_i)]
    \label{obj_unitary}
\end{equation}
Using the Arrow-Pratt approach one obtains,
\begin{equation*}
\alpha^\star_i\approx \frac{\mu}{-\frac{U_i''(W)}{U_i'(W)}W\sigma^2}
\end{equation*}
Under CARA utility and normally distributed excess return, the relationship is exact: \(\alpha^\star_i = \frac{\mu}{A_iW\sigma^2}\)
and under CRRA utility we get \(\alpha^\star_i \approx \frac{\mu}{\gamma_i\sigma^2}\).

It is common to introduce constant weights \(\lambda_i>0\) which capture the extent to which individual preferences are represented in the household's objective. Since full commitment is assumed, they can be thought of as exogenous weights capturing an unmodeled bargaining process. Henceforth, I assume
\[\lambda_1+\lambda_2=1\]
In this paper, I discuss a number of alternative approaches to using these weights to aggregate individual preferences.

\section{Classical Collective Approach}
\label{section_Utilitarian}
A common approach to modeling the intra-household portfolio choice is to assume that the household maximizes a weighted sum of expected utilities:
\begin{equation}
    \max_{\alpha} \; \mathbb{E} [U_{h}(C_h)] = \max_{\alpha} \; \sum_i \lambda_i \mathbb{E} [U_i(C_h)] 
    \label{obj_sum_public}
\end{equation}
This approach is intuitively appealing because by varying \(\lambda\)'s one will span the whole set of Pareto-efficient outcomes (\cite{vermeulen2002collective}). Here the weight \(\lambda_i\) is referred to as the Pareto weight of individual \(i\). The criterion in (\ref{obj_sum_public}) is utilized by a number of papers including \cite{serra2022risk, yilmazer2015portfolio, gu2024gender, Inkmann_Michaelides_Zhang_2026}. However, as I discuss next, its application to the portfolio choice problem suffers from major shortcomings.

\subsection{Individual Risk Aversion and Non-Monotonicity of the Household Portfolio Choice}
\label{subsection_NonMonotonicity}

In this subsection, I demonstrate that the classical collective approach summarized by equation (\ref{obj_sum_public}) yields a counterintuitive relationship between the risk aversion of a spouse and the share of risky assets chosen by the household. Namely, as shown by the following proposition, the share of household's portfolio invested in the risky asset may increase with the risk aversion of its member.

\begin{proposition}
\label{prop_nonmonotonicity_sum}
     Let \(\theta = (\mu,\sigma,r_f,W,\lambda_1,\lambda_2) \). Then, for both CARA and CRRA individual utilities, there exists \(\theta  \in \mathbb{R}_{+}^6 \) and a risk aversion level of member \(j\) such that solution \(\alpha^\star\) to the household problem in (\ref{obj_sum_public}) is a non-monotonic function of the risk aversion of member \(i\).
\end{proposition}

\begin{proof}
    In the Appendix Section \ref{sec_AppProofs}.
\end{proof}

Figure \ref{fig_nonmonotonicity_sum} uses the two examples from the proof to illustrate the result. Optimal household risky share is plotted against individually optimal ones which serve as a reference.
\begin{figure}[thb]
\centering
\begin{subfigure}{0.45\textwidth}
    \includegraphics[width=1\linewidth]{"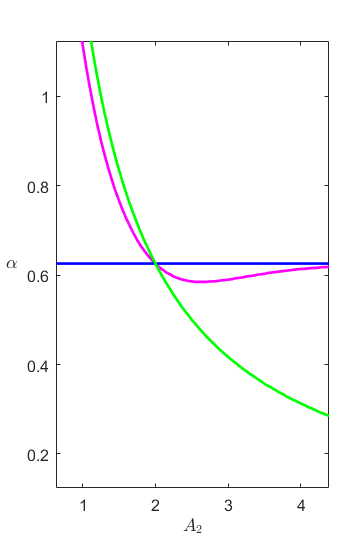"}
    \caption{CARA}
\end{subfigure}
\begin{subfigure}{0.45\textwidth}
    \includegraphics[width=1\linewidth]{"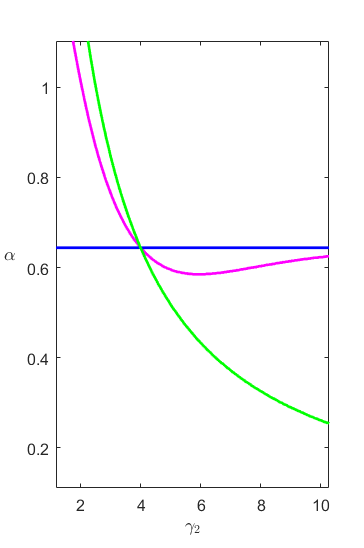"}
    \caption{CRRA}
\end{subfigure}

    \caption{\footnotesize 
    This figure plots the optimal risky share \(\alpha\), against the risk aversion of the second member of the household for CARA (Panel A) and CRRA (Panel B) individual utility functions. Blue curves plot \(\alpha^\star_1\) and green curves plot \(\alpha^\star_2\) which are the solutions to the single's problem and are defined in equation (\ref{obj_unitary}). The magenta curves depict the household-level optimal risky share \(\alpha^\star\), which solves the objective in (\ref{obj_sum_public}). Set, \(\mu= 0.1, \ \sigma=0.2, \ r_f=0, \ W=2,\) and \(\lambda_1=\lambda_2=0.5\). In Panel A set \(A_1=2\) and in Panel B set \(\gamma_1=4\). There are 10,000 simulations of a normally distributed shock and consumption is always positive in practice. 
    }
    \label{fig_nonmonotonicity_sum}
\end{figure}
As can be seen from the figure, the optimal household risky share initially decreases with the risk aversion of the second member but eventually starts sloping upwards even though the risk aversion continues to increase. Thus, monotonicity of the household risky share with respect to the risk aversion of a member cannot be guaranteed under collective approach whereby individual CARA or CRRA utilities are aggregated additively. This is a rather puzzling result: increasing the risk aversion of a member can increase the risk-taking of the household. I refer to this counter-intuitive pattern as non-monotonicity. It implies that such a collective framework would not be reliable in a setting where one is interested in risk preference heterogeneity. For example, one immediate practical implication is an identification problem: for a given level of risk aversion of the first member, different levels of risk aversion of the second member may be consistent with the observed behavior of the household. 

Figure \ref{fig_nonmonotonicity_sum} also demonstrates that such counterintuitive result occurs under realistic values of \(\mu,\sigma,r_f,\) and individual risk aversion parameters. Section \ref{sec_AppVariations} of the Appendix shows that non-monotonicity arises under other values of exogenous variables including the distributional assumptions of the risky asset return. In particular, it also shows that household risk-taking can increase with the risk aversion of the spouse even if that spouse has a higher bargaining weight.

\subsubsection{Intuition and the Binary Choice Problem}

Intuition for the non-monotonicity result can be obtained by noting that increasing the risk aversion of an agent decreases her marginal utility scale on a certain consumption domain\footnote{More precisely, for all positive consumption levels under CARA and for consumption levels above 1 under CRRA.}. Therefore, using constant Pareto weights the household may effectively downweight the preference for a lower risky share of the more risk averse spouse whose risk aversion increases further. This results in the household solution approaching a higher level preferred by the less risk averse spouse. 

I now introduce the following Lemma which will then help better understand why the optimal household risky share may increase with the risk aversion of a spouse. Let \(M_i\) be the derivative of the individual expected utility of spouse \(i\) at the optimal risky share with respect to the risky share:
\begin{equation}
    M_i(\alpha^\star) \equiv \frac{\partial \mathbb{E} (U_{i}(\alpha^\star))}{\partial \alpha}
\end{equation}

\begin{lemma}
\label{lemma_IFT}
For both CARA and CRRA individual utilities, solution \(\alpha^\star\) to the household problem in (\ref{obj_sum_public}) is a decreasing function of the risk aversion of member \(i\) if and only if \(M_i(\alpha^\star)\) is a decreasing function of the risk aversion of member \(i\).
\end{lemma}

\begin{proof}
    In the Appendix Section \ref{sec_AppProofs}.
\end{proof}

So, the non-monotonicity result of Proposition \ref{prop_nonmonotonicity_sum} is related to the derivative of the individual expected utility at the optimal risky share with respect to the risky share. If this derivative is a non-monotonic function of individual risk aversion at the optimal risky share, then the household risky share will be a non-monotonic function of the risk aversion of this individual. In other words, collective solution \(\alpha^\star\) is decreasing in the risk aversion of the second member as long as \(\partial M_2(\alpha^\star)/\partial A_2<0 \) in the case of CARA and \(\partial M_2(\alpha^\star)/\partial \gamma_2<0\) in the case of CRRA. Note that the risk aversion of member \(j\) is still relevant because it affects \(\alpha^\star\) which in turn affects \(M_i(\alpha^\star)\).

I now return to the binary choice example provided in the introduction and highlight the parallels to the continuous choice set-up of this section. I provide intuition for the result in Proposition \ref{prop_nonmonotonicity_sum}, by building on Lemma \ref{lemma_IFT}. Let \(s\) be the risky alternative and \(f\) be the safe alternative. The household prefers the risky alternative whenever

\begin{equation}
\label{def_Delta_h}
\Delta_h^s \equiv \mathbb{E}[U_h(s)] - U_h(f)  = \lambda_1 \mathbb{E}[U_1(s)]+\lambda_2 \mathbb{E}[U_2(s)] - \lambda_1 U_1(f)- \lambda_2 U_2(f)>0. 
\end{equation}
Assume for a given level of risk aversion of the first member, there is a level of risk aversion of the second member, such that the household prefers the risky alternative for any risk aversion of the second member below this level and prefers the safe alternative for some risk aversion of the second member just above this level. Then, the analog of the monotonicity concept under the continuous choice set-up would be maintaining the preference for the safe alternative as the risk aversion of the second member increases above this threshold level. For an increasing risk aversion of the second member, the household will never switch back to preferring the risky alternative if,

\begin{equation}
\label{ineq_monotone_binary}
\frac{\partial\Delta_h^s}{\partial A_2} = \lambda_2 \left( \frac{\partial \mathbb{E}[U_2(s)]}{\partial A_2} - \frac{\partial  U_2(f)}{\partial A_2} \right) \equiv \lambda_2\frac{\partial \Delta_2^s}{\partial A_2}  < 0
\end{equation}
Note the analogy between the \(\frac{\partial \Delta_2^s}{\partial A_2}\) above and the \(\frac{\partial M_2(\alpha^\star)}{\partial A_2}\) used in Lemma \ref{lemma_IFT}. Figure \ref{fig_nonmonotonicity_sum_binary} plots \(\Delta_h^s\) for an example in which inequality in (\ref{ineq_monotone_binary}) is violated, resulting in non-monotonicity.

\begin{figure}[ht]
\centering
\begin{subfigure}{0.45\textwidth}
    \includegraphics[width=1\linewidth]{"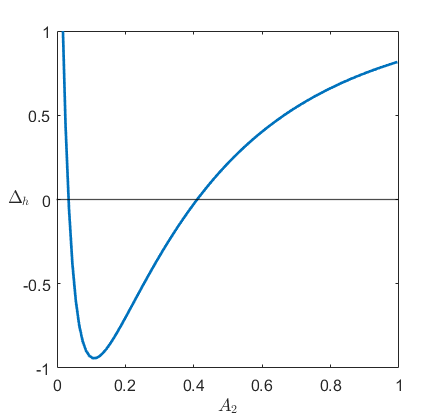"}
    \caption{CARA}
\end{subfigure}
\begin{subfigure}{0.45\textwidth}
    \includegraphics[width=1\linewidth]{"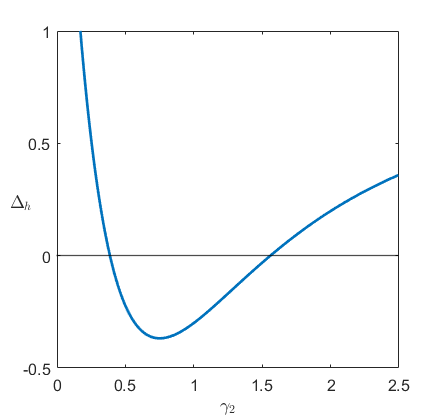"}
    \caption{CRRA}
\end{subfigure}
    \caption{\footnotesize 
    This figure plots the \(\Delta_h^s\) which represents the relative preference towards the risky alternative against the risk aversion of the second member of the household. \(\Delta_h^s\) is defined precisely in equation (\ref{def_Delta_h}). Assume the risky alternative \(s\) provides \$1 with probability 0.9 and \$60 with probability 0.1. Assume the safe alternative \(f\) provides a certain \$5. This example is taken from \cite{apesteguia2018monotone}. In Panel A set \(A_1=0.005\) and in Panel B set \(\gamma_1=0.08\).    }
    \label{fig_nonmonotonicity_sum_binary}
\end{figure}
This demonstrates that non-monotonicity arises because increasing the risk aversion of a member may make them dislike the risky alternative less even though they still prefer the safe alternative. That is, the individual utility difference between the risky and safe gambles may be a non-monotonic function of the risk aversion of this individual. An analogous dynamic takes place in the continuous distribution of risk and continuous choice set-up: as agent 2 becomes more risk averse, she prefers a lower risky share but she also becomes more indifferent to similar risky shares, making the household objective more similar to the objective of agent 1. 

The patterns in the binary choice example can be linked to the findings of \cite{apesteguia2018monotone}. The authors study the random utility model (RUM) under the unitary framework and show that the RUM can be a non-monotonic function of the risk aversion. However, in part 1 of Proposition 1, the authors show that monotonicity of RUM is directly related to the monotonicity of \(\mathbb{E}[U_i(s)] - U_i(f)\) which, as I have argued above, is in turn directly related to the monotonicity of the household-level decision\footnote{In the language of \cite{apesteguia2018monotone}, \(s\) and \(f\) are not \(\Omega\)-ordered in my model.}.

\subsubsection{Conditions that Guarantee Monotonicity}

The following proposition specifies a restrictive condition under which monotonicity can be guaranteed under CARA preferences. 

\begin{proposition}
\label{prop_monotonicity_CARA}
     Solution \(\alpha^\star\) to the household problem in (\ref{obj_sum_public}) is always a decreasing function of \(A_i\) under CARA utility functions when \(A_i \leq A_j\).
\end{proposition} 

\begin{proof}
    In the Appendix Section \ref{sec_AppProofs}.
\end{proof}

So, under CARA preferences, the household risky share always decreases with the risk aversion of the less risk averse member. In Figure \ref{fig_nonmonotonicity_sum}, this corresponds to the region to the left of the point where all three curves intersect. There is, however, no equivalent result under CRRA individual utilities: Figure \ref{fig_nonmonotonicitybelow_sum_CRRA} plots an example where household risky share increases with the risk aversion of the less risk averse member.

\begin{figure}[th]
\centering
    \includegraphics[width=0.5\linewidth]{"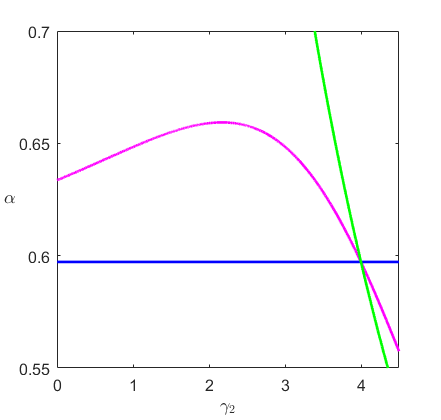"}
    \caption{\footnotesize 
    This figure plots the optimal risky share \(\alpha\), against the risk aversion of the second member of the household for CRRA individual utility functions. The blue curve plots \(\alpha^\star_1\) and the green curve plots \(\alpha^\star_2\) which are the solutions to the single's problem and are defined in equation (\ref{obj_unitary}). The magenta curve depicts the household-level optimal risky share \(\alpha^\star\), which solves the objective in (\ref{obj_sum_public}). Set, \(\mu= 0.1, \ \sigma=0.2, \ r_f=0, \ W=0.5, \ \lambda_1=\lambda_2=0.5\) and \(\gamma_1=4\). There are 10,000 simulations of a normally distributed shock and consumption is always positive in practice. 
    }
    \label{fig_nonmonotonicitybelow_sum_CRRA}
\end{figure}

Finally, one might notice a somewhat arbitrary choice of the wealth level \(W=2\) in the example given in Figure \ref{fig_nonmonotonicity_sum}. This is because under \(W=1\), one always obtains monotonic household risky share under CRRA preferences. The following proposition shows formally that this is always true.

\begin{proposition}
\label{prop_monotonicity_sum_CRRA}
    Solution \(\alpha^\star\) to the household problem in (\ref{obj_sum_public}) is a strictly decreasing function of the risk aversion of member \(i\) under CRRA utility function when \(r_f=0\), \(W=1\).
\end{proposition}

\begin{proof}
    In the Appendix Section \ref{sec_AppProofs}.
\end{proof}

The intuition for Proposition \ref{prop_monotonicity_sum_CRRA} can be obtained by noting that the marginal CRRA individual utility of consumption at \(C_i=1\) is equal to 1 and is independent of the risk aversion parameter. The discussion in Appendix Section \ref{sec_AppWealthSplit} elaborates on this point based on a wealth splitting example. Under CARA utility function however, there is no positive level of consumption that would make the marginal utility independent of the risk aversion. Appendix Section \ref{sec_AppWealthSplit} offers an alternative formulation of the portfolio choice problem whereby shocks affect the wealth additively and the initial wealth is set to zero. This generates a result under CARA utility function similar to Proposition \ref{prop_monotonicity_sum_CRRA} due to the fact that risk aversion does not change the marginal CARA utility of consumption at \(C_i=0\).

\subsection{Wealth Sensitivity of the Household Portfolio Choice}
\label{subsection_WealthSensitivity}
In this section, I discuss the role played by the wealth level \(W\) in the household portfolio choice, and highlight further concerns with applications of the classical collective framework. In the unitary framework, the wealth level has no effect on the portfolio choice \(\alpha\) under CRRA preferences and on the dollar amount of risky investment \(\alpha W\) under CARA preferences. I show that this independence breaks in the classical collective framework.

\begin{figure}[htb]
\centering
\begin{subfigure}{0.45\textwidth}
    \includegraphics[width=1\linewidth]{"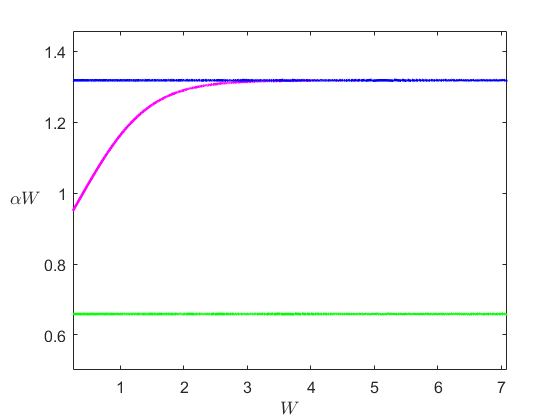"}
    \caption{CARA, \(\lambda_1=0.5\)}
\end{subfigure}
\begin{subfigure}{0.45\textwidth}
    \includegraphics[width=1\linewidth]{"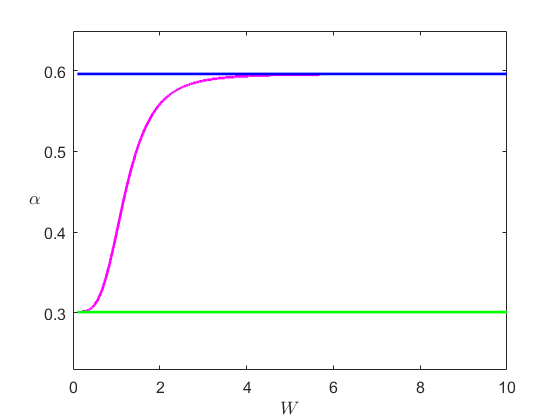"}
    \caption{CRRA, \(\lambda_1=0.5\)}
\end{subfigure}
\begin{subfigure}{0.45\textwidth}
    \includegraphics[width=1\linewidth]{"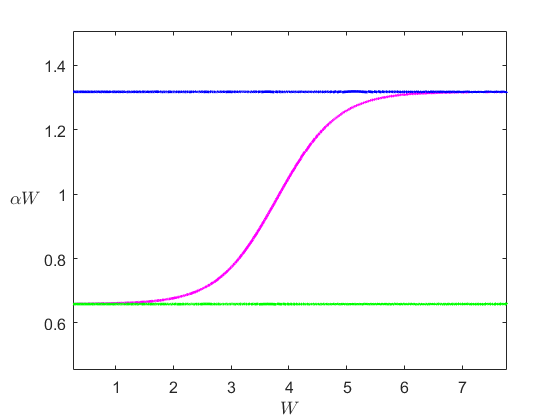"}
    \caption{CARA, \(\lambda_1=0.001\)}
\end{subfigure}
\begin{subfigure}{0.45\textwidth}
    \includegraphics[width=1\linewidth]{"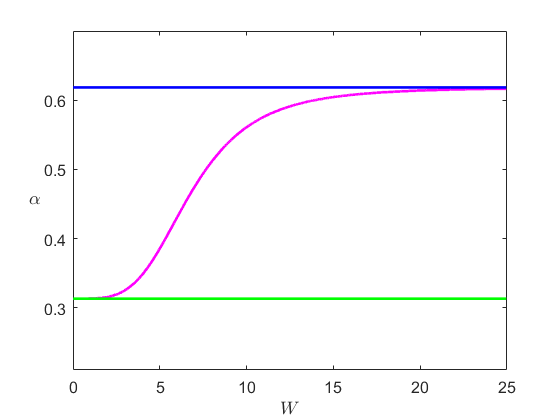"}
    \caption{CRRA, \(\lambda_1=0.001\)}
\end{subfigure}
    \caption{\footnotesize 
        This figure illustrates the optimal risk-taking behavior (risky investment \(\alpha W\) for CARA in Panels A and C and the optimal share \(\alpha\) for CRRA in Panels B and D), against the wealth level of the household (\(W\)). Blue and green curves plot solutions to the single's problem defined in equation (\ref{obj_unitary}). The magenta curves depict household-level solutions. Set, \(\mu= 0.1, \ \sigma=0.2, \ r_f=0,\). In Panels A and C set \(A_1=2, \ A_2=4\) and in Panels B and D set \(\gamma_1=4, \ \gamma_2=8\). In Panels A and B set \(\lambda_1 = 0.5\) and in Panels C and D set \(\lambda_1=0.001\).  There are 10,000 simulations of a normally distributed shock and consumption is always positive in practice. 
    }
    \label{fig_Desmos_CARA_sum_W}
\end{figure}

Figure \ref{fig_Desmos_CARA_sum_W} plots the optimal risky share for CRRA and the optimal dollar amount of risky investment for CARA against the level of wealth. The optimal risk-taking increases with the wealth level\footnote{The paper by \cite{neelakantan2013portfolio} shows that the risk-taking increases with the wealth level of the household in a similar collaborative framework (Proposition 1). However, the wealth sensitivity of the risky share highlighted here is distinct from the one discussed by the authors. This is due to the fact that the results in \cite{neelakantan2013portfolio} rely on an arbitrary parameter \(\delta\) which essentially rescales the Pareto weight (and hence the marginal utility) of the second member (higher delta means lower weight placed on the utility of the second member). This \(\delta\) is reverse engineered such that for a given level of wealth \(\gamma^{hh}=\frac{\gamma_1+\gamma_2}{2}\). Hence, \(\delta\) depends on the wealth level and the difference in risk aversion parameters. For example, under equal Pareto weights, when the second member is the more risk averse one, higher wealth would require a lower \(\delta\) which would make the solutions more similar to the individual solution of the second member. I note that, this approach is not micro-founded and in some sense monotonicity is artificially enforced. Authors do not show what happens to the household solution when \(\delta\) is fixed.}. More importantly, under both CARA and CRRA, as wealth levels increase (decrease), the household-level solution converges rapidly to the individual solution of the less (more) risk averse individual. Under reasonable model parametrization and CRRA preferences, for a sufficiently high wealth level, the household solution closely resembles the solution of the less risk averse individual even using a very low Pareto weight for that individual (\(\lambda_1=0.001\)). Similarly, for a sufficiently low wealth level, the household solution closely resembles the solution of the more risk averse individual even using a very low Pareto weight for that individual (\(\lambda_2=0.001\)). This switching pattern makes application of a collective decision making framework whereby individual utility functions are aggregated additively, problematic in the context of portfolio choice and heterogeneous household wealth.

The patterns described above can be attributed to the cause of the non-monotonicity result discussed in Section \ref{subsection_NonMonotonicity}. This is because there is a threshold consumption level above which the marginal utility of the more risk averse individual is smaller than that of the more risk tolerant one and below which the marginal utility of the more risk averse individual is greater than that of the more risk tolerant one. Hence, under greater wealth levels, the household maximizing the sum of individual utilities is going to effectively downweight the preferences of the more risk averse spouse and the solution will resemble the one preferred by the more risk tolerant spouse. Similarly, under lower wealth levels, the household will favor the preferences of the more risk averse spouse.

Another relevant concern is related to the mutual identification of the wealth level and the Pareto weights of individuals. As can be seen from Figure \ref{fig_Desmos_CARA_sum_W}, more than one combination of \(\lambda_1\) and \(W\) can produce the same optimal risky share \(\alpha^\star\). This is not an issue if one observes either the wealth level of the household or the bargaining weights of individual members and is willing to use it directly in the model. If, however, both are unobserved, then the model will be unable to pin down both the wealth level and the Pareto weights.

\subsection{Literature Implications}
\label{section_Literature}

In this subsection, I review the existing literature in light of the findings of Sections \ref{subsection_NonMonotonicity} and \ref{subsection_WealthSensitivity} and provide a reconciliation with some of the counterintuitive results.

\cite{gu2024gender} is a recent empirical work interested in bargaining asymmetry in household portfolio choice decisions. Using Australian survey data, authors structurally estimate bargaining weights and find a large gender gap: ``in an average household, the husband’s bargaining power is 60\%, whereas the wife’s is 40\%". The authors also discuss gender and non-gender based determinants of this gap.

In order to estimate these bargaining weights, the authors rely on the same setup as the one described in Section \ref{section_Setup}. In particular their setup is also static, with full spousal commitment and consumption of a single public good. Authors assume CARA individual utilities and normal distribution of shocks. Although the authors allow for participation costs, these costs can be set to zero in order to reconcile their approach with the results of Sections \ref{subsection_NonMonotonicity} and \ref{subsection_WealthSensitivity}. In \cite{gu2024gender}, the authors assume the household first aggregates risk aversion parameters of its members by a harmonic average (\(1/A_h=\lambda_1/A_1+\lambda_2/A_2\)) and then solves the standard individual maximization problem. Using my notation, this results in the following risky share solution. 

\begin{equation}
    \alpha^\star = \frac{\mu}{ \frac{1}{\frac{\lambda_1}{A_1}+\frac{\lambda_2}{A_2} } W\sigma^2}
            \label{sol_GPZ}
\end{equation}
Using the CARA expression for \(\alpha_i^\star\), the above equation can be re-written as,
\[
\alpha^\star=\lambda_1\alpha^\star_1+\lambda_2\alpha^\star_2
\]
This household solution is monotonically decreasing in the individual risk aversion levels. Surprisingly, however, authors claim that their modeling approach is equivalent to an approach whereby the household maximizes individual expected utilities aggregated additively:

\begin{quote}
\footnotesize
Under the conditions detailed in Section B of the Internet Appendix, this gives an equivalent expression as in the classical collective bargaining model in which the household utility function is a weighted average of the individual’s utility (Manser and Brown 1980; McElroy and Horney 1981; Chiappori 1988a, 1992). Therefore, in our model, aggregating individual risk preferences is equivalent to aggregating individual utility functions.
\end{quote}

In Section \ref{subsection_NonMonotonicity}, I have shown that aggregating individual CARA utilities by a weighted average results in a non-monotonic solution which contradicts the equivalence claim made by \cite{gu2024gender}. In the Appendix Section \ref{sec_AppOtherApp}, I show that using Mean-Variance or simplified Mean-Variance individual utility functions also results in solutions different from the one obtained by the authors. Indeed, this equivalence claim is not accurate. Authors' argument in Section B of the Internet Appendix essentially amounts to showing that under certain assumptions (in particular when \(1/A_h = \lambda_1/A_1+\lambda_2/A_2\)):
\begin{equation*}
U_h(\alpha^\star) = \lambda_1U_1(\alpha^\star_1) + \lambda_2U_2(\alpha^\star_2)
\end{equation*}
This, however, does not imply:
\begin{equation*}
U_h(\alpha) = \lambda_1U_1(\alpha) + \lambda_2U_2(\alpha) \; \; \forall \alpha
\end{equation*}

Note that the solution in (\ref{sol_GPZ}) seems intuitively appealing. For example, it can be motivated by an argument that the household chooses a weighted arithmetic average of the two individually optimal risky shares (i.e., \(\alpha^\star = \lambda_1\alpha_1^\star + \lambda_2\alpha_2^\star\)). It is also the same as the risky share obtained by first splitting the wealth according to the bargaining weights, investing it separately as dictated by individual preferences, and then combining it back. However, using harmonic average of individual risk aversion parameters or, equivalently  arithmetic average of the individually optimal risky shares results in a solution not motivated by the classical theory of collective bargaining. As I will show in the next section, a more micro-founded approach results in a solution using the arithmetic average of individual risk aversion parameters or, equivalently, the harmonic average of the individually optimal risky shares.

Another empirical work on intra-household portfolio choices is a paper by \cite{yilmazer2015portfolio} which uses the HRS data. The authors do not impose functional form assumptions on the individual utility functions yet rely on the same framework. In particular, it is assumed that household maximizes the weighted sum of expectations. The authors claim that their model predicts that the risky share will increase as the risk aversion of a spouse decreases.

\begin{quote}
\footnotesize
    One testable hypothesis derived from the theoretical framework is that the household invests a larger proportion of its wealth in risky assets when the spouse with more bargaining power is more risk tolerant.
\end{quote}

However, this claim is not supported by the theoretical framework used in the paper and it is inaccurate. For example, we have seen in Section \ref{subsection_NonMonotonicity} that for both CARA and CRRA individual utility functions the household can decrease its investments in the risky asset when one of the spouses gets more risk tolerant (less risk averse), even when that spouse has greater bargaining power. This demonstrates that monotonicity might have been previously falsely taken for granted in the literature.
    
\cite{thornqvist2014bargaining} study participation, risky share and diversification decisions using administrative data from Sweden. The authors show empirically that an increase in female bargaining power (who are more risk averse on average) results in household taking less risk. The authors rely on a similar theoretical framework with full spousal commitment and consumption of a single public good. In their model, authors have two periods of consumption and the shocks have a binary distribution. The household is assumed to maximize a sum of expected CRRA utilities. Since authors do not directly utilize the model in their empirical analysis the solution to the household problem is not discussed in the main text. In particular, non-monotonicity implicit in their framework is not emphasized.

Finally, a paper by \cite{mazzocco2004saving} is an influential paper interested in the precautionary savings behavior of couples with heterogeneous risk aversion and prudence. Unlike in the framework described in Section \ref{section_Setup}, the author assumes that the household can only invest in a risk-free asset and each spouse separately consumes their own private good. The uncertainty in the model comes only from the stochastic labor income and saving less implies taking more risk. Author obtains the following unexpected theoretical result: household savings can decrease when one spouse becomes more risk averse. My paper complements this work by showing that this counterintuitive result can be obtained with the public consumption assumption. In addition, \cite{mazzocco2004saving} proves a theorem on non-monotonicity: increasing the risk aversion of the \textit{less} risk averse member can increase the risk tolerance of the household. This is in contrast to my proposition \ref{prop_monotonicity_CARA} which shows that increasing the risk aversion of the less risk averse member can only decrease the risky share of the household under CARA preferences.

\section{Equal Marginal Risk Premia Solution}
\label{section_ERP}
In this section I propose an alternative approach to modeling the intra-household behavior under uncertainty whereby the household trades-off individual marginal risk premia. I show that this approach results in a monotonic household portfolio choice and obeys other intuitive properties.

In order to describe the solution, I first consider the dollar amount agent \(i\) is willing to sacrifice out of her initial wealth in order to marginally reduce the risky share. More precisely, let \(\hat{\pi}_i(\epsilon)\) be defined by the following condition

\begin{equation}
\label{def_pieps}
    V_i(W-\hat{\pi}_i(\epsilon),\alpha-\epsilon)=V_i(W,\alpha)
\end{equation}
where \( V_i(W,\alpha) \equiv \mathbb{E} [U_i(W(1+\alpha\tilde{x}))]\). The difference between \(\hat{\pi}_i(\epsilon)\) and the classical definition of risk premium (i.e. \(V_i(W-risk \ premium,0)=V_i(W,\alpha)\)) is that the risk is allowed to have a positive mean and we are not reducing the risk to zero (in other words, \(\epsilon<\alpha\)). Instead we are considering a marginal reduction in the risky share. Now, let \({\pi}_i\) be the marginal risk premium of agent \(i\). More precisely,
\begin{equation*}
    \pi_i\equiv\frac{\partial \hat{\pi}_i(\epsilon)}{\partial \epsilon} 
\end{equation*}
We can now define \(\alpha^\star\) by the following condition:

\begin{equation}
\label{def_alphastar}
    \lambda_1\pi_1(\alpha^\star)+\lambda_2\pi_2(\alpha^\star)=0
\end{equation}
Appendix Section \ref{sec_AppERP} shows that alternative definitions of risk premia in terms of dollar amounts or relative marginal amounts would result in equivalent conditions defining the household risky share. I refer to this approach summarized by equation (\ref{def_alphastar}) as equal marginal risk premia solution.

Now, by differentiating equation (\ref{def_pieps}) with respect to \(\epsilon\), one obtains 
\begin{equation}
\label{def_pi}
    \pi_i= -\frac{\partial V_i}{\partial \alpha} /\frac{\partial V_i}{\partial W}
\end{equation}
In the Appendix Section \ref{sec_AppERP}, I show that under small risk and small \(\alpha\) or small \(\alpha-\alpha_i^\star\), \(\pi_i\) can be approximated by:
\begin{equation}
\label{approx_pi}
    \pi_i \approx -W(\mu +W\frac{U_i''(W)}{U_i'(W)}\sigma^2\alpha) \approx  -W^2\frac{U_i''(W)}{U_i'(W)}\sigma^2(\alpha-\alpha^\star_i)
\end{equation}

\subsection{Interpretation}
In this subsection I am going to provide interpretation of the equations in (\ref{def_alphastar}), (\ref{def_pi}) and (\ref{approx_pi}) for a general utility function.

First, note that since \(\frac{\partial V_i}{\partial W}>0\), equation (\ref{def_pi}) implies that the sign of \(\pi_i\) is equal to the sign of \(-\frac{\partial V_i}{\partial \alpha}\). Assuming the F.O.C. of the individual maximization problem is satisfied we get that \(\pi_i=0\) at \(\alpha=\alpha^\star_i\). That is, at an individually optimal risky share, the individual is not willing to sacrifice any wealth in order to alter the risky share. Assuming the S.O.C. of the individual maximization problem is satisfied (i.e. \(\frac{\partial^2 V_i}{\partial \alpha^2}<0\)), when \(\alpha<\alpha^\star_i\), we have \(\pi_i<0\). That is, the individual would require compensation in order to decrease the risky share further away from his preferred level. Finally,  when \(\alpha>\alpha^\star_i\), we have \(\pi_i>0\). That is, the individual is willing to reduce its initial wealth in order to decrease the risky share. This relationship between \(\pi_i\) and \(\alpha-\alpha^\star_i\) is also captured by the approximation in (\ref{approx_pi}).

Now, assume spouses have heterogeneous risk preferences and person 1 is the more risk tolerant and person 2 is the more risk averse: \(\alpha_1^\star>\alpha_2^\star\) and consider an arbitrary \(\alpha\) such that \(\alpha_1^\star>\alpha>\alpha_2^\star\). Then, \(\pi_2(\alpha)>0>\pi_1(\alpha)\). If spouses are homogeneous in their risk preferences, equation (\ref{def_alphastar}) is satisfied if and only if \(\alpha=\alpha_1^\star=\alpha_2^\star\). Then, \(\pi_1(\alpha)=\pi_2(\alpha)=0\). Note that for \(\alpha>\alpha_1^\star\), we have \(\pi_1(\alpha)>0\) and \(\pi_2(\alpha)>0\) and equation (\ref{def_alphastar}) cannot be satisfied as \(\lambda_1\pi_1+\lambda_2\pi_2>0\). Likewise, for \(\alpha<\alpha_2^\star\), we have \(\pi_1(\alpha)<0\) and \(\pi_2(\alpha)<0\) and equation (\ref{def_alphastar}) cannot be satisfied as \(\lambda_1\pi_1+\lambda_2\pi_2<0\). Hence, \(\alpha^\star\) defined by condition (\ref{def_alphastar}) is Pareto efficient (i.e. \(\alpha_i^\star\geq\alpha^\star\geq\alpha_j^\star\) for \(\alpha_i^\star\geq\alpha_j^\star\)).

Differentiate equation (\ref{def_alphastar}) with respect to some parameter of the second member's utility function \(\psi_2\) (for example, \(A_2\) under CARA or \(\gamma_2\) under CRRA):

\begin{equation*}
    \lambda_1\left(\frac{\partial \pi_1}{\partial \alpha}\frac{\partial\alpha^\star}{\partial \psi_2} \right) + \lambda_2\left(\frac{\partial\pi_2}{\partial \psi_2}+ \frac{\partial \pi_2}{\partial \alpha}\frac{\partial\alpha^\star}{\partial \psi_2} \right) = 0
\end{equation*}
\begin{equation}
\label{compstat_erp}
    \frac{\partial\alpha^\star}{\partial \psi_2} = -\frac{\lambda_2\frac{\partial\pi_2}{\partial \psi_2}}{\lambda_1 \frac{\partial \pi_1}{\partial \alpha} + \lambda_2 \frac{\partial \pi_2}{\partial \alpha}}
\end{equation}
So, the relationship between the household risky share and the risk aversion parameter depends on the relationship between the marginal risk premium and risk aversion parameter and the relationship between the marginal risk premium and the level of the risky share. Equation (\ref{compstat_erp}) will be useful below when deriving comparative statics for CARA and CRRA preferences.

Assuming \(\lambda_1=\lambda_2\), the condition in (\ref{def_alphastar}) can be interpreted as follows. The amount the more risk tolerant spouse needs to be compensated by for a marginal decrease in the risky share is equal to the amount the more risk averse spouse is willing to pay for a marginal decrease the risky share. The weights \(\lambda_i\) in equation (\ref{def_alphastar}) can be interpreted as follows. Suppose \(\lambda_2/\lambda_1=2\). This means that spouse 1 is willing to accept twice as mush as spouse 2 is willing to compensate for a change in the risky share. In other words, their weighted risk premiums have to be equal in absolute terms. Below I show that increasing the bargaining weight of a spouse always moves the risky share towards the one preferred by that spouse.

Equation (\ref{def_alphastar}) can be justified as follows. Assume person 1 is the more risk tolerant and person 2 is the more risk averse: \(\alpha_1^\star>\alpha_2^\star\) and consider an initial or candidate household risky share \(\alpha=\alpha_1^\star\). At this level, in order to decrease the risky share by a small amount, spouse 2 would be willing to pay person 1 more than person 1 would require for such a change. Further decrease in \(\alpha\) would take place until \(\alpha=\alpha^\star\) at which point in order to decrease the risky share further, person 1 would require a greater compensation than person 2 is willing to make. This reasoning can be applied to any  level of the risky share and the household will always arrive at the same \(\alpha^\star\) for a given value of exogenous parameters.

\subsection{Estimation}

In this section I am going to discuss CARA and CRRA marginal risk premia and the resulting household portfolio choice. For simplicity  I assume \(r_f=0\). Assuming risky return is normally distributed, under CARA utility function \(\pi_i\) is given by (see Appendix Section \ref{sec_AppERP}): 
\begin{equation}
\label{def_pi_cara}
    \pi_i =\frac{-W(\mu-\alpha A_iW\sigma^2)}{1+\alpha(\mu-\alpha A_iW\sigma^2)}
\end{equation}
The condition in (\ref{def_alphastar}) can then be written as:

\begin{equation}
\label{def_alphastar_cara}
    \lambda_1\frac{\mu-\alpha^\star A_1W\sigma^2}{1+\alpha^\star(\mu-\alpha^\star A_1W\sigma^2)}+\lambda_2\frac{\mu-\alpha^\star A_2W\sigma^2}{1+\alpha^\star(\mu-\alpha^\star A_2W\sigma^2)}=0
\end{equation}
The magenta curve in Panel A of Figure \ref{fig_erp_base} plots the economically relevant solution to equation (\ref{def_alphastar_cara}). Under CRRA utility function \(\pi_i\) is given by (see Appendix Section \ref{sec_AppERPCRRA}): 
\begin{equation}
\label{def_pi_crra}
   \pi_i = -W\frac{\mathbb{E}[\tilde{x}(1+\alpha\tilde{x})^{-\gamma_i}]}{\mathbb{E}[(1+\alpha\tilde{x})^{1-\gamma_i}]}
\end{equation}
and the condition in (\ref{def_alphastar}) can then be written as:
\begin{equation}
\label{def_alphastar_crra}
    \lambda_1\frac{\mathbb{E}[\tilde{x}(1+\alpha^\star\tilde{x})^{-\gamma_1}]}{\mathbb{E}[(1+\alpha^\star\tilde{x})^{1-\gamma_1}]}+\lambda_2\frac{\mathbb{E}[\tilde{x}(1+\alpha^\star\tilde{x})^{-\gamma_2}]}{\mathbb{E}[(1+\alpha^\star\tilde{x})^{1-\gamma_2}]} = 0
\end{equation}

\begin{figure}[thb]
\centering
\begin{subfigure}{0.45\textwidth}
    \includegraphics[width=1\linewidth]{"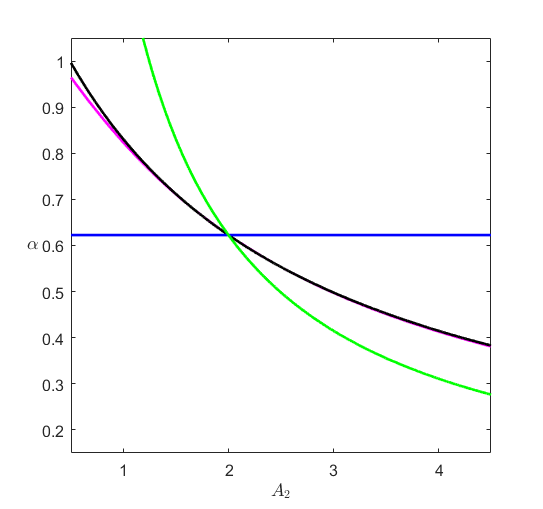"}
    \caption{CARA}
\end{subfigure}
\begin{subfigure}{0.45\textwidth}
    \includegraphics[width=1\linewidth]{"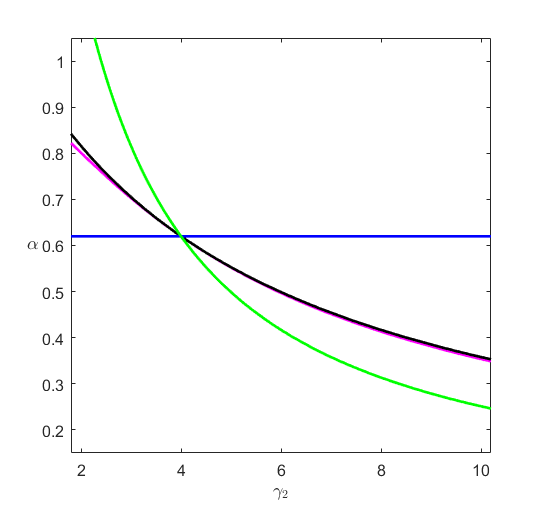"}
    \caption{CRRA}
\end{subfigure}

    \caption{\footnotesize 
    This figure plots the optimal risky share \(\alpha\), against the risk aversion of the second member of the household for CARA (Panel A) and CRRA (Panel B) individual utility functions. Blue curves plot \(\alpha^\star_1\) and green curves plot \(\alpha^\star_2\) which are the solutions to the single's problem and are defined in equation (\ref{obj_unitary}). The magenta curves depict the household-level optimal risky share \(\alpha^\star\), which solves the objective in (\ref{def_alphastar}). The black curves depict the harmonic average solution.  Set, \(\mu= 0.1, \ \sigma=0.2, \ r_f=0, \ W=2,\) and \(\lambda_1=\lambda_2=0.5\). In Panel A set \(A_1=2\) and in Panel B set \(\gamma_1=4\). There are 10,000 simulations of a normally distributed shock and consumption is always positive in practice. 
    }
    \label{fig_erp_base}
\end{figure}

Note that \(\alpha^\star\) does not depend on the wealth level under CRRA. The magenta curve in Panel B of Figure \ref{fig_erp_base} plots the economically relevant solution to (\ref{def_alphastar_crra}) using a numerical evaluation.  

As can be seen from the Figure, the optimal risky share decreases with the risk aversion of the spouses under the CARA and CRRA examples used in Proposition \ref{prop_nonmonotonicity_sum}. The following proposition argues that the equal marginal risk premia solution is always decreasing with the risk aversion of a spouse.

\begin{proposition}
\label{prop_monotonicity_erp}
     For both CARA and CRRA individual utilities, solution \(\alpha^\star\) to the household problem in (\ref{def_alphastar}) is a strictly decreasing function of the risk aversion of its member.
\end{proposition}

\begin{proof}
    In the Appendix Section \ref{sec_AppERP}.
\end{proof}

So, unlike in the classical collective approach, household risky share is a monotonic function of the risk aversion of a spouse under the equal marginal risk premia approach. 

Another intuitive property for the household solution is for it to always approach the risky share preferred by a member as the bargaining weight of that member increases.  The following proposition shows that this is true for the equal marginal risk premia solution.

\begin{proposition}
\label{prop_bargainingerp}
Assume \(\alpha_1^\star>\alpha^\star>\alpha_2^\star\), then \(\frac{\partial \alpha^\star}{\partial \lambda_1}>0\) for both CARA and CRRA individual utilities.
\end{proposition}

\begin{proof}
In the Appendix Section \ref{sec_AppERP}.
\end{proof}

Next, I am going to show that the equal marginal risk premia solution is well approximated by the harmonic average of the two individually optimal levels of the risky share. Equivalently, the effective household risk aversion is well approximated by the arithmetic average of the individual risk aversion parameters. 

Using the approximation in (\ref{approx_pi}), the condition in (\ref{def_alphastar}) can be written as:
\begin{equation*}
    \lambda_1\left(\mu +W\frac{U_1''(W)}{U_1'(W)}\sigma^2\alpha^\star\right)+\lambda_2\left(\mu +W\frac{U_2''(W)}{U_2'(W)}\sigma^2\alpha^\star\right)\approx0 
\end{equation*}
Which yields:
\begin{equation*}
      \alpha^\star \approx \frac{\mu}{-\left(\lambda_1\frac{U_1''(W)}{U_1'(W)}+\lambda_2\frac{U_2''(W)}{U_2'(W)}\right)W\sigma^2}
\end{equation*}
Under CARA preferences, the above equation can then be written as,
\begin{equation}
\label{sol_harmonic_CARA}
    \alpha^\star \approx \frac{\mu}{(\lambda_1A_1+\lambda_2A_2)W\sigma^2} = \frac{1}{\frac{\lambda_1}{\alpha_1^\star}+\frac{\lambda_2}{\alpha_2^\star}}
\end{equation}
and under CRRA preferences,
\begin{equation}
\label{sol_harmonic_CRRA}
    \alpha^\star\approx  \frac{\mu}{(\lambda_1\gamma_1+\lambda_2\gamma_2)\sigma^2} \approx \frac{1}{\frac{\lambda_1}{\alpha_1^\star}+\frac{\lambda_2}{\alpha_2^\star}}
\end{equation}
The black curves in Figure \ref{fig_erp_base} plot the approximations in (\ref{sol_harmonic_CARA}) and (\ref{sol_harmonic_CRRA}). As can be seen from the figure, the true solution is indeed well approximated by the harmonic average of the two individually optimal levels of the risky share. This approximation improves when the two risk aversion parameters get closer to each other and when risk aversion of a spouse increases. According to the approximation of the equal marginal risk premia solution, the effective household risk aversion is a weighted \textit{arithmetic average} of the individual risk aversions. Note that this is in contrast to the approach proposed by \cite{gu2024gender} in which the household risk aversion is a weighted \textit{harmonic average} of the individual risk aversions.

\subsection{Equivalences}
In this section I demonstrate approaches to the intra-household decision making which result in a solution equivalent to the one obtained above. Namely, I consider the collective approach with normalized weights and expected product of utilities approach.

\subsubsection{Normalized Collective Approach}
Rewrite the equal marginal risk premia condition in (\ref{def_alphastar}) as:

\begin{equation}
\label{def_alphastar_ERP_rewrite}
    \lambda_1\frac{\frac{\partial V_1(\alpha^\star)}{\partial \alpha} }{\frac{\partial V_1(\alpha^\star)}{\partial W}} +\lambda_2\frac{\frac{\partial V_2(\alpha^\star)}{\partial \alpha} }{\frac{\partial V_2(\alpha^\star)}{\partial W}} = 0
\end{equation}
or equivalently,
\begin{equation}
\label{def_alphastar_ERP_rewrite_2}
    \lambda_1\frac{\partial V_1(\alpha^\star)}{\partial \alpha}  +\frac{\lambda_2}{\delta^\star}\frac{\partial V_2(\alpha^\star)}{\partial \alpha}  = 0
\end{equation}
where \(\delta^\star=\frac{\partial V_2(\alpha^\star)}{\partial W}/\frac{\partial V_1(\alpha^\star)}{\partial W}\). Note that if \(\delta^\star=1\), the condition above is equivalent to the F.O.C. of the utilitarian approach discussed in Section \ref{section_Utilitarian}. Otherwise, the condition in (\ref{def_alphastar_ERP_rewrite}) demonstrates that if Pareto weights of the utilitarian approach are divided by the marginal expected utility of wealth at the optimal risky share then such criteria results in the same solution as the one produced by the equal marginal risk premia approach discussed in Section \ref{section_ERP}. Note however, that such approach is not well formulated as the objective function would depend on the solution itself.

Note that under the small risk assumption, we have \(\frac{\partial V_i(\alpha^\star)}{\partial W} \approx \frac{\partial U_i(W)}{\partial W} \), and let
\[\delta=\frac{\partial U_2(W)}{\partial W}/\frac{\partial U_1(W)}{\partial W}\]
Now, consider an alternative household risky share defined by the following condition:

\begin{equation}
    \label{def_alphastar_normutilitarian}
       \lambda_1\frac{\partial V_1(\alpha^\star)}{\partial \alpha}  +\frac{\lambda_2}{\delta}\frac{\partial V_2(\alpha^\star)}{\partial \alpha}  = 0
\end{equation}
That is, \(\delta^\star\) in the equal marginal risk premia condition (\ref{def_alphastar_ERP_rewrite_2}) is replaced by \(\delta\). I refer to the condition in (\ref{def_alphastar_normutilitarian}) as normalized collective approach. The red curves in Figure \ref{fig_erp_equiv} depict a solution to (\ref{def_alphastar_normutilitarian}) under CARA and CRRA utility functions. The Figure demonstrates that the equal marginal risk premia solution is indeed well approximated by the solution produced by the normalized collective approach. Here the approximation is best when \(\lambda_1=\lambda_2=0.5\). Note that when \(r_f > 0\)  the definition of \(\delta\) needs to be modified.

Appendix Section \ref{sec_AppWealthSplit} shows that imposing a certain intuitive requirement on the optimal wealth share in a simple wealth splitting problem results in a normalization of Pareto weights similar to the normalization in equation (\ref{def_alphastar_normutilitarian}). Also note that Proposition \ref{prop_monotonicity_sum_CRRA} is relevant to this equivalence because under CRRA utility function and \(W=1\), we get \(\delta=\delta^\star=1\) and the classical collective approach discussed in Section \ref{section_Utilitarian} is equivalent to the normalized collective approach and the equal marginal risk premia approach which, as discussed in Section \ref{section_ERP}, always results in a monotonic solution. Appendix Section \ref{sec_AppWealthSplit} also discusses an equivalent result for CARA.

\begin{figure}[thb]
\centering
\begin{subfigure}{0.45\textwidth}
    \includegraphics[width=1\linewidth]{"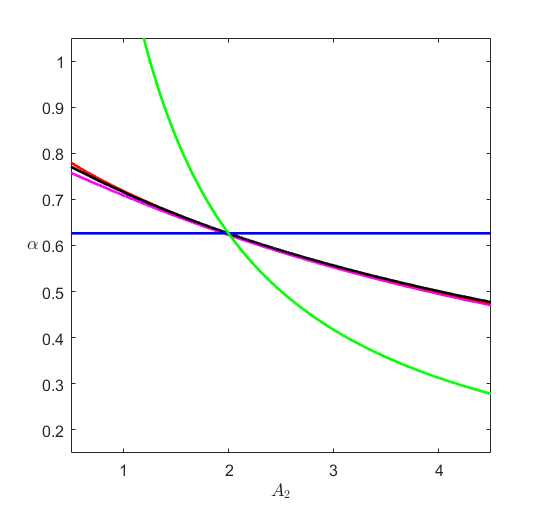"}
    \caption{CARA}
\end{subfigure}
\begin{subfigure}{0.45\textwidth}
    \includegraphics[width=1\linewidth]{"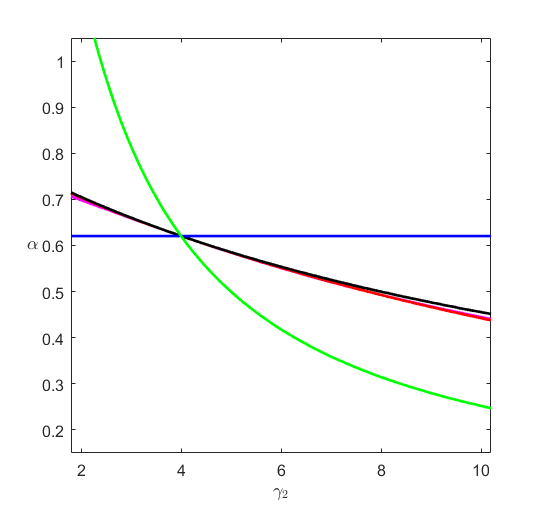"}
    \caption{CRRA}
\end{subfigure}

    \caption{\footnotesize 
    This figure plots the optimal risky share \(\alpha\), against the risk aversion of the second member of the household for CARA (Panel A) and CRRA (Panel B) individual utility functions. Blue curves plot \(\alpha^\star_1\) and green curves plot \(\alpha^\star_2\) which are the solutions to the single's problem and are defined in equation (\ref{obj_unitary}). The magenta curves depict the household-level optimal risky share \(\alpha^\star\), which solves the objective in (\ref{def_alphastar}). The red curves depict the normalized utilitarian solution. The black curves depict the harmonic average solution.  Set, \(\mu= 0.1, \ \sigma=0.2, \ r_f=0, \ W=2,\) and \(\lambda_1=0.75\) and \(\lambda_2=0.25\). In Panel A set \(A_1=2\) and in Panel B set \(\gamma_1=4\). There are 10,000 simulations of a normally distributed shock and consumption is always positive in practice. 
    }
    \label{fig_erp_equiv}
\end{figure}

\subsubsection{Expected product of utilities}
\label{subsection_expprod}
An alternative to aggregating individual utilities additively is to aggregate them multiplicatively. More precisely, assume the household maximizes the expected product of utilities\footnote{Due to the fact that utilities might be negative, in practice, we multiply them by -1 and minimize the objective instead of maximizing it. For the sake of completeness I also consider the product of expected utilities criteria in the Appendix Section \ref{sec_AppOtherApp}.}:
\begin{equation}
\label{obj_expprod}
    \max_{\alpha} \; \mathbb{E}[U_{h}(C_h)] = \max_{\alpha} \; \mathbb{E} \left[ \prod_{i} U_i(C_h)^{\lambda_i} \right]
\end{equation}
Appendix Section \ref{sec_AppExpProd} shows that multiplying individual CARA or CRRA utilities in this way yields a household objective equivalent to the respective individual ones. More precisely, under CARA preferences, one obtains:
\begin{equation*}
    \alpha^\star =    \frac{\mu}{(\lambda_1A_1+\lambda_2A_2)W\sigma^2} = \frac{1}{\frac{\lambda_1}{\alpha_1^\star}+\frac{\lambda_2}{\alpha_2^\star}}
\end{equation*}
and under CRRA preferences and the small risk assumption, one obtains:
\begin{equation*}
    \alpha^\star \approx \frac{\mu}{(\lambda_1\gamma_1+\lambda_2\gamma_2)\sigma^2} \approx \frac{1}{\frac{\lambda_1}{\alpha_1^\star}+\frac{\lambda_2}{\alpha_2^\star}}
\end{equation*} 
Note that these solutions are identical to the harmonic average approximations obtained in Section \ref{section_ERP}. Appendix Section \ref{sec_AppOtherApp} shows that this is also the solution to the the sum of simplified Mean-Variance utilities criteria.

In contrast to the other papers discussed in Section \ref{section_Literature}, a paper by \cite{addoum2017household} relies on the same criteria as in (\ref{obj_expprod}). The author studies the portfolio choice of married couples in HRS as they transition into retirement. The author explicitly acknowledges that preferences can be heterogeneous and that the bargaining power may shift. The author uses a similar framework except there are three periods of consumption. Individuals are assumed to have CRRA preferences and the shocks follow a log-normal distribution.

\section{Conclusion}
\label{section_Conclusion}

I analyze application of the collective model to the household portfolio choice problem and highlight its major shortcomings. In particular, I show that household risk-taking may increase when one of the spouses gets more risk averse\footnote{Despite the fact that household formation is not studied in the current theoretical framework, it can be argued that perfect matching on risk preferences cannot be guaranteed if risk aversion is heterogeneous across genders in aggregate.}. Since household risk-taking also decreases with the risk aversion of the spouse, I refer to this pattern as non-monotonicity. I provide some precise characterization of this phenomenon and demonstrate that it occurs under realistic calibration of the model. In addition, I demonstrate that irrespective of the bargaining weights, for lower levels of wealth, the household closely mimics the more risk averse member and for slightly higher levels of wealth the household closely mimics the less risk averse member. These issues pose a major challenge for applications of the collective approach to the portfolio choice problem.

I review the recent empirical literature which has explicitly or implicitly relied on the given approach to study intra-household financial risk-taking and highlight some potential confusion. A recent paper by \cite{gu2024gender} proposes a preference aggregation approach which results in a solution free of the non-monotonicity problem. However, I argue that authors falsely claim that this approach is equivalent to the classical collective bargaining model. Authors essentially propose that the household optimal risky share is a weighted \textit{arithmetic average} of the individually optimal risky shares. Equivalently, the effective household risk aversion is a weighted harmonic average of the individual risk aversions. I argue that although seemingly intuitive, the proposed solution is not micro-founded. 

Given the issues with the application of the classical collective approach and the lack of a well-motivated alternative, I offer an approach to modeling household portfolio choice using the concept of risk premium. The household is assumed to choose a risky share which equates the dollar amount the more risk-tolerant spouse is requiring for a marginal decrease in the risky share to the dollar amount the more risk-averse spouse is willing to pay for this decrease. I show that this results in a Pareto-efficient solution which obeys other intuitive properties. I also show that this household risky share is approximated by the weighted \textit{harmonic average} of the individually optimal ones. Equivalently, the effective household risk aversion is a weighted arithmetic average of the individual risk aversions. Finally, I discuss two alternative approaches to the household portfolio choice problem that are approximately equivalent to the proposed equal marginal risk premia approach. First is a collective approach whereby Pareto weights are scaled by the inverses of marginal utility of wealth. Second, is an approach used in \cite{addoum2017household} whereby the household maximizes the expected product of utilities.

This paper offers a reliable theoretical framework for the future empirical research interested in intra-household decision-making under uncertainty and estimation of the gender asymmetry in bargaining. The proposed approach achieves desired theoretical properties not merely by relying on an alternative consumption specification or other assumptions of the collective model, but by appropriate fundamental characterization of the problem.

\newpage
\printbibliography

\newpage
\section*{Appendix}

\renewcommand{\thesubsection}{\Alph{subsection}}

\subsection{Classical Collective Approach}
\label{sec_AppProofs}

\subsubsection*{Proposition \ref{prop_nonmonotonicity_sum}}

\begin{proof}
I prove the proposition by providing an example for each individual utility specification. Let us start with the case of CARA. Here, individual utility function can be written as,
    \begin{equation}
        U^{cara}_i = \frac{-exp(-A_iW(1+r_f+\alpha \tilde{x}))}{A_i}
    \end{equation}
Hence, because for any \(\tilde{y} \sim \mathcal{N}(\mu, \sigma^2)\), we have \(    \mathbb{E}[exp(\tilde{y})] = exp \left(\mu + \frac{\sigma^2}{2} \right)     \label{expproperty} \), individual expected utility can be written as

\begin{equation}
\label{def_ExpectedUtility_CARA}
    \mathbb{E} (U^{cara}_i) =  \frac{-exp \left(-A_iW(1+r_f+\alpha\mu) +\frac{A_i^2W^2\alpha^2\sigma^2}{2} \right)}{A_i}
\end{equation}
Then, the objective in (\ref{obj_sum_public}) can be written as,
\begin{equation*}
\begin{split}
    \lambda_1 \frac{-exp \left(-A_1W(1+r_f+\alpha\mu) +\frac{A_1^2W^2\alpha^2\sigma^2}{2} \right)}{A_1} +    \\ \lambda_2 \frac{-exp \left(-A_2W(1+r_f+\alpha\mu) +\frac{A_2^2W^2\alpha^2\sigma^2}{2} \right)}{A_2}
\end{split}
\end{equation*}
and the F.O.C. is:

\begin{equation}
\begin{split}
    \lambda_1 exp \left(-A_1W(1+r_f+\alpha^\star\mu) +\frac{A_1^2W^2{\alpha^\star}^2\sigma^2}{2} \right)(-\mu + \alpha^\star A_1W \sigma^2)  + \\
    \lambda_2 exp \left(-A_2W(1+r_f+\alpha^\star\mu) +\frac{A_2^2W^2{\alpha^\star}^2\sigma^2}{2} \right)(-\mu + \alpha^\star A_2W \sigma^2) = 0
\end{split}
\label{FOC_sum_CARA}
\end{equation}
There is no closed form solution to equation (\ref{FOC_sum_CARA}) but it can be solved for numerically. Let \(\mu= 0.1, \ \sigma=0.2, \ r_f=0, \ W=2, \  \lambda_1=\lambda_2=0.5\) and \(A_1=2\). The optimal risky share \(\alpha^\star\) which solves the equation (\ref{FOC_sum_CARA}) is approximately equal to 0.62, 0.58, 0.61, for \(A_2=2\), \(A_2=2.5\) and \(A_2=4\) respectively. We have found an example where \(\alpha^\star\) is a non-monotonic function of the risk aversion of a member under CARA preferences.

Let us now consider the case of CRRA. Here, individual utility can be written as,
\begin{equation}
   U^{crra}_i =  \frac{\left[W(1+r_f+\alpha\tilde{x})\right]^{1-\gamma_i}}{1-\gamma_i}
\end{equation}
and the individual expected utility can be written as,

\begin{equation}
\label{def_ExpectedUtility_CRRA}
    \mathbb{E} (U^{crra}_i) =  \frac{\mathbb{E}\left(\left[W(1+r_f+\alpha\tilde{x})\right]^{1-\gamma_i}\right)}{1-\gamma_i}
\end{equation}
Unlike in the case of CARA, there is no equivalent expression for the objective in (\ref{obj_sum_public}) or the F.O.C. which does not include an expectation operator. Let \(\mu= 0.1, \ \sigma=0.2, \ r_f=0, \ W=2, \ \lambda_1=\lambda_2=0.5\) and \(\gamma_1=4\). I solve for \(\alpha^\star\) directly using a Monte-Carlo simulation with 10,000 normally distributed shocks. The consumption always stays positive in practice. The optimal risky share \(\alpha^\star\) which maximizes the objective in (\ref{obj_sum_public}) for CRRA preferences is approximately equal to 0.62, 0.57 and 0.6, for \(\gamma_2=4\), \(6\) and \(10\) respectively. Since we found an example for both CARA and CRRA where \(\alpha^\star\) is a non-monotonic function of the risk aversion of a member we have completed the proof.\end{proof}

\subsubsection*{Lemma \ref{lemma_IFT}}

\begin{proof}

Using the implicit function theorem one can write:
\begin{equation}
    \frac{\partial \alpha^\star}{\partial A_2} = - \frac{\frac{\partial}{\partial A_2}\partial (V(\alpha^\star,A_2)/\partial \alpha) }{\frac{\partial}{\partial \alpha}\partial (V(\alpha^\star,A_2)/\partial \alpha)}
\label{equation_IFT}
\end{equation}
where \(V\) is the household objective function \(\sum_i \lambda_i \mathbb{E} [U_i(C_h)]\) evaluated at the solution \(\alpha^\star\) and \(\partial (V(\alpha^\star,A_2)/\partial \alpha)=\lambda_1M_1(\alpha^\star)+\lambda_2M_2(\alpha^\star) \). The denominator of the RHS of equation (\ref{equation_IFT}) can thus be written as:
\begin{equation}
\label{equation_IFT_denom}
\begin{split}
    \frac{\partial}{\partial \alpha}\partial (V(\alpha^\star,A_2)/\partial \alpha) = \lambda_1\frac{\partial M_1(\alpha^\star)}{\partial \alpha} + \lambda_2\frac{\partial M_2(\alpha^\star)}{\partial \alpha} 
\end{split}
\end{equation}
Since the S.O.C. of the individual maximization problem is always satisfied, both terms of the RHS of equation (\ref{equation_IFT_denom}) are negative under CARA and CRRA. Hence, the sign of the LHS is the same as the sign of the numerator which is equal to \(\lambda_2\frac{\partial M_2(\alpha^\star)}{\partial A_2}\)
\end{proof}

Below, I show the S.O.C. for the CARA utility function as an example (the arguments are analogous for the CRRA utility function). Dropping the star superscript from the optimal \(\alpha\):

\begin{equation}
    M_i(\alpha^\star) = - exp \left(-A_iW(1+r_f+\alpha\mu) +\frac{A_i^2W^2\alpha^2\sigma^2}{2} \right)(-W\mu + \alpha A_i W^2\sigma^2)
\end{equation}

\begin{equation*}
\begin{split}
    \frac{\partial M_i(\alpha^\star)}{\partial \alpha} = \\ 
    - exp \left(-A_iW(1+r_f+\alpha\mu) +\frac{A_i^2W^2\alpha^2\sigma^2}{2} \right) \times (-A_iW\mu + \alpha A_i^2 W^2\sigma^2) \times \\ (-W\mu + \alpha A_iW^2 \sigma^2) \\ 
   - exp \left(-A_iW(1+r_f+\alpha\mu) +\frac{A_i^2W^2\alpha^2\sigma^2}{2} \right) \times (A_iW^2 \sigma^2) < 0 \\
    - (-A_iW\mu + \alpha A_i^2W^2 \sigma^2) \times (-W\mu + \alpha A_i W^2 \sigma^2)  - A_i W^2\sigma^2 < 0 \\
    -(\alpha A_i W \sigma^2 -\mu)^2  - \sigma^2 < 0
\end{split}
\end{equation*}

\subsubsection*{Proposition \ref{prop_monotonicity_CARA}}
I first provide the following Lemma which shows that the optimal risky share of the household is always between the two individually optimal counterfactual risky shares under CARA preferences.
\begin{lemma}
\label{lemma_inbetween}
     Let \(\alpha^\star_i=\frac{\mu}{A_iW\sigma^2}\) and assume \(\lambda_i>0\) for \(i \in \{1,2\}\). Then, under CARA utility, \(\alpha^\star= \alpha_1^\star=\alpha_2^\star\) when \(\alpha_1^\star=\alpha_2^\star\), and \(\alpha_i^\star<\alpha^\star< \alpha_j^\star\) when \(\alpha_i^\star<\alpha_j^\star\).
\end{lemma} 

\begin{proof}
    Note that, since exponents are always positive, one of the following must be true for the F.O.C. in equation (\ref{FOC_sum_CARA}) to hold:
\begin{equation*}
    \begin{split}
        (-W\mu + \alpha^\star A_1 W^2\sigma^2) =       (-W\mu + \alpha^\star A_2 W^2\sigma^2) = 0 \\
        (-W\mu + \alpha^\star A_2 W^2\sigma^2) < 0 <   (-W\mu + \alpha^\star A_1 W^2 \sigma^2) \\
        (-W\mu + \alpha^\star A_1 W^2\sigma^2) < 0 <   (-W\mu + \alpha^\star A_2 W^2 \sigma^2)
    \end{split}
\end{equation*}
In the first case, we get \(\alpha^\star=\alpha^\star_1=\alpha^\star_2\).
In the second case, we get \(\alpha_1^\star<\alpha^\star< \alpha_2^\star\) and in the third case we get \(\alpha_2^\star<\alpha^\star< \alpha_1^\star\).  \end{proof}

\begin{proof}
Without loss of generality, assume \(i=2\). Given Lemma \ref{lemma_IFT}, we need to show that \(\frac{\partial M_2(\alpha^\star)}{\partial A_2}<0\)
    \begin{equation}
\begin{split}
    \frac{\partial M_2(\alpha)}{\partial A_2} = \\ 
    - exp \left(-A_2W(1+r_f+\alpha\mu) +\frac{A_2^2W^2\alpha^2\sigma^2}{2} \right) \times \\ (-W(1+r_f+\alpha\mu) + A_2W^2\alpha^2\sigma^2) \times  (-W\mu + \alpha A_2W^2 \sigma^2) \\ 
     - exp \left(-A_2W(1+r_f+\alpha\mu) +\frac{A_2^2W^2\alpha^2\sigma^2}{2} \right) \times (\alpha W^2\sigma^2) < 0
\end{split}
\end{equation}
Note that by Lemma \ref{lemma_inbetween}, when \(\mu>0\) we have \(\alpha^\star>0\) since \(\alpha_i^\star>0\) for \(i \in \{1,2\}\). Hence, one can rewrite the condition as,
\begin{equation}
\begin{split}
 \label{ineq_IFT_CARA}
    - [-W(1+r_f+\alpha\mu) + A_2W^2\alpha^2\sigma^2] \times (-W\mu + \alpha A_2 W^2 \sigma^2) - \alpha W^2\sigma^2 < 0 \\
    [\alpha A_2W^2\sigma^2-W\mu -W(1+r_f)/\alpha] \times (\alpha A_2 W^2 \sigma^2-W\mu) + W^2\sigma^2 > 0 \\
    [\alpha A_2W\sigma^2-\mu -(1+r_f)/\alpha] \times (\alpha A_2 W \sigma^2-\mu) + \sigma^2 > 0 \\
    \left[\alpha -\mu/A_2W\sigma^2 -\frac{1+r_f}{\alpha A_2W\sigma^2} \right] \times (\alpha -\mu/A_2 W \sigma^2) + 1/A_2^2W^2\sigma^2 > 0 \\
    \left[\alpha^\star - \alpha_2^\star -\frac{\alpha_2^\star(1+r_f)}{\alpha^\star \mu} \right] \times (\alpha^\star - \alpha_2^\star) + 1/A_2^2W^2\sigma^2 > 0\\
    \left[(\alpha^\star - \alpha_2^\star\left(1 +\frac{(1+r_f)}{\alpha^\star \mu} \right)\right] \times (\alpha^\star - \alpha_2^\star) + 1/A_2^2W^2\sigma^2 > 0
\end{split}
\end{equation}

This inequality always holds due to Lemma \ref{lemma_inbetween} as \(\alpha^\star<\alpha_2^*\) for all \(A_2<A_1\) (as \(\alpha_1^\star<\alpha_2^\star\)) and \(\alpha^\star=\alpha_2^*\) for all \(A_2=A_1\) (as \(\alpha_1^\star=\alpha_2^\star\))\footnote{Note that the second term of the LHS in the inequality above is 0 in the individual case when \(\alpha^\star=\mu/A_2W\sigma^2=\alpha_2^\star\). So the inequality always holds in the individual case for \(\mu>0\).
}.\end{proof}

\subsubsection*{Proposition \ref{prop_monotonicity_sum_CRRA}}

\begin{proof}
    Given Lemma \ref{lemma_IFT}, we need to show that \(\frac{\partial M_i(\alpha^\star)}{\partial \gamma_i}<0\). Given the definition in equation (\ref{def_ExpectedUtility_CRRA}), individual expected utility under CRRA can be written as:
    
    \begin{equation}
    \mathbb{E} (U_i) =  \frac{W^{1-\gamma_i}}{1-\gamma_i}
    \mathbb{E}\left(\left[(1+r_f+\alpha\tilde{x})\right]^{1-\gamma_i}\right)
    \end{equation}
    Then, one can write:
   \begin{equation}
   \begin{split}
        M_i(\alpha) \equiv \frac{\partial \mathbb{E} (U_{i}(\alpha))}{\partial \alpha} = 
        \frac{W^{1-\gamma_i}}{1-\gamma_i}
    \mathbb{E}\left( \frac{\partial\left[(1+r_f+\alpha\tilde{x})\right]^{1-\gamma_i}}{\partial \alpha} \right) = \\
    W^{1-\gamma_i} \mathbb{E} \left[ \tilde{x}(1+r_f+\alpha \tilde{x})^{-\gamma_i} \right]
    \end{split}
    \end{equation}
   Set \(W=1,r_f=0\). We can then write,
    \begin{equation}
    \begin{split}
        \frac{\partial M_i(\alpha)}{\partial \gamma_i} = \frac{\partial \mathbb{E} \left[ \tilde{x}(1+\alpha \tilde{x})^{-\gamma_i} \right] }{\partial \gamma_i} =  \mathbb{E} \left[ \frac{\partial \left (\tilde{x}(1+\alpha \tilde{x})^{-\gamma_i}\right) }{\partial \gamma_i} \right] = \\
        \mathbb{E}\left[ -\tilde{x}ln(1+\alpha\tilde{x})(1+\alpha\tilde{x})^{-\gamma_i} \right]
    \end{split}
    \end{equation}
    Let \(\tilde{y} \equiv 1+\alpha\tilde{x}\). Hence, \(\tilde{x}=\frac{\tilde{y}-1}{\alpha}\) and we can rewrite the above equality as 

    \begin{equation}
        \frac{\partial M_i(\alpha^\star)}{\partial \gamma_i} = -\frac{1}{\alpha^\star}\mathbb{E}\left[ (\tilde{y}-1)ln(\tilde{y})(\tilde{y})^{-\gamma_i} \right]
    \end{equation}
     Since the term \((\tilde{y}-1)ln(\tilde{y})\) is strictly positive for all non-zero \(\tilde{y}\neq1\) and is equal to zero for \(\tilde{y}=1\), we obtain \(\frac{\partial M_i(\alpha^\star)}{\partial \gamma_i}<0\).
\end{proof}

\subsubsection*{Alternative Parametrization of the Classic Collective Approach}
\label{sec_AppVariations}

Examples below illustrate that non-monotonicity arises under model parametrization different from the one considered in Section \ref{section_Utilitarian}. Figure \ref{fig_sum_uniform_Matlab} plots the optimal risky shares assuming uniformly distributed risky asset return.

\begin{figure}[ht]
\centering
\begin{subfigure}{0.45\textwidth}
    \includegraphics[width=\linewidth]{"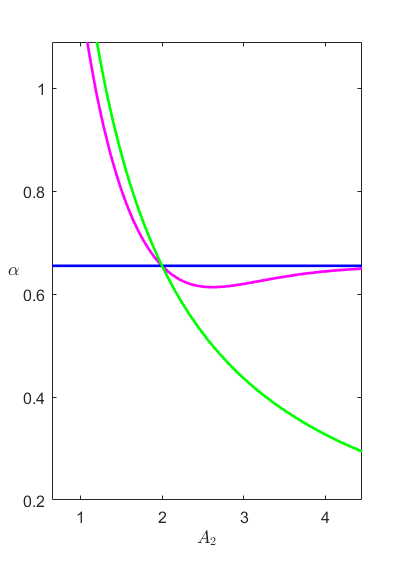"}
    \caption{CARA}
\end{subfigure}
\begin{subfigure}{0.45\textwidth}
    \includegraphics[width=\linewidth]{"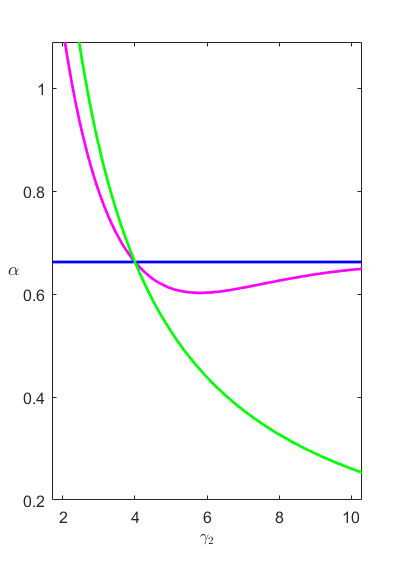"}
    \caption{CRRA}
\end{subfigure}
    \caption{\footnotesize
        This figure plots the optimal risky share \(\alpha\), against the risk aversion of the second member of the household for CARA (Panel A) and CRRA (Panel B) individual utility functions assuming  \(\tilde{x} \sim \mathcal{U}(-0.25, 0.45)\). Blue curves plot \(\alpha^\star_1\) and green curves plot \(\alpha^\star_2\) which are the solutions to the single's problem and are defined in equation (\ref{obj_unitary}). The magenta curves depict the household-level optimal risky share \(\alpha^\star\), which solves the objective in (\ref{obj_sum_public}). Set, \(r_f=0, \ W=2,\) and \(\lambda_1=\lambda_2=0.5\). In Panel A set \(A_1=2\) and in Panel B set \(\gamma_1=4\).}
    \label{fig_sum_uniform_Matlab}
\end{figure}

Figure \ref{fig_sum_variations_Matlab} plots the optimal risky shares for alternative levels of bargaining weights and risk aversion.

\begin{figure}[ht]
\centering
\begin{subfigure}{0.4\textwidth}
    \includegraphics[width=\linewidth]{"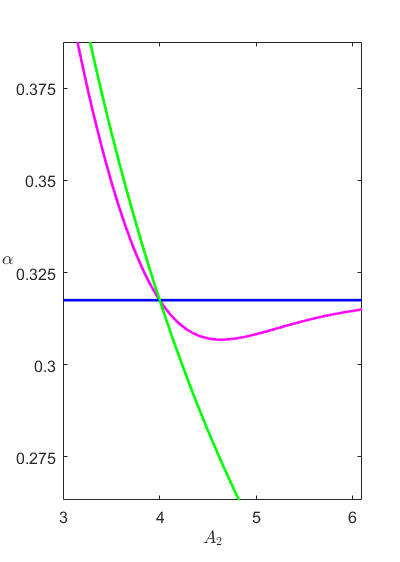"}
    \caption{CARA, \(A_1 = 4\)}
\end{subfigure}
\begin{subfigure}{0.4\textwidth}
    \includegraphics[width=\linewidth]{"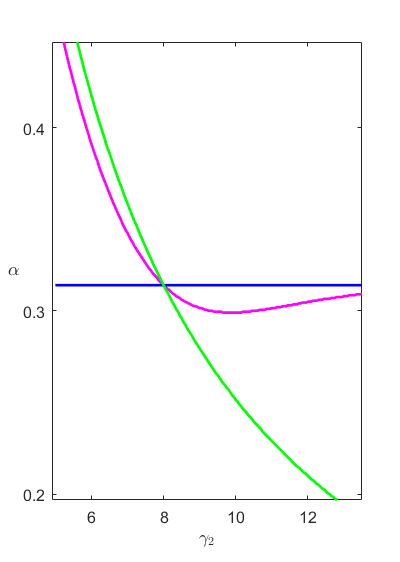"}
    \caption{CRRA, \(\gamma_1=8\)}
\end{subfigure}
\begin{subfigure}{0.4\textwidth}
    \includegraphics[width=\linewidth]{"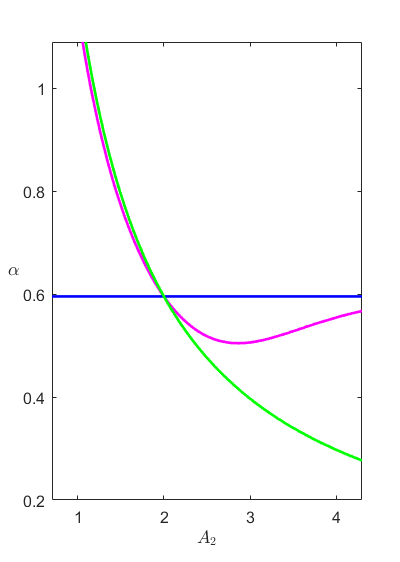"}
    \caption{CARA, \(\lambda_1=0.2\)}
\end{subfigure}
\begin{subfigure}{0.4\textwidth}
    \includegraphics[width=\linewidth]{"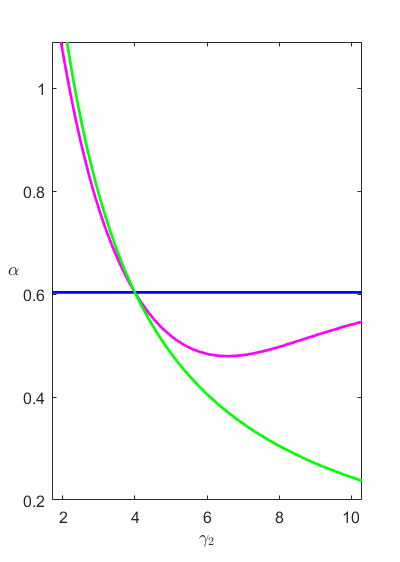"}
    \caption{CRRA, \(\lambda_1=0.2\)}
\end{subfigure}
    \caption{\footnotesize
    This figure plots the optimal risky share \(\alpha\), against the risk aversion of the second member of the household for CARA (Panels A and C) and CRRA (Panels B and D) individual utility functions under alternative levels of the first agent's risk aversion (Panels A and B) and bargaining weights (Panels C and D). Blue curves plot \(\alpha^\star_1\) and green curves plot \(\alpha^\star_2\) which are the solutions to the single's problem and are defined in equation (\ref{obj_unitary}). The magenta curves depict the household-level optimal risky share \(\alpha^\star\), which solves the objective in (\ref{obj_sum_public}). Set, \(\mu= 0.1, \ \sigma=0.2, \ r_f=0,\) and  \( W=2\). In Panels A and B set \(A_1=4\), and \(\gamma_1=8\) respectively and set \(\lambda_1=\lambda_2=0.5\). In Panels C and D set \(A_1=2\) and \(\gamma_1=4\) respectively and set \(\lambda_1=0.2\) and \(\lambda_2=0.8\).}
    \label{fig_sum_variations_Matlab}
\end{figure}

\subsection{Equal Marginal Risk Premia}
\label{sec_AppERP}
\subsubsection*{Alternative Definitions of Risk Premia}

Note that although in Section \ref{section_ERP} one can potentially define \(\alpha^\star\) in terms of \(\hat{\pi}_1(\epsilon)\) and \(\hat{\pi}_2(\epsilon)\) in equation (\ref{def_alphastar}) and solve directly for \(\hat{\pi}_i(\epsilon)\), I use the marginal risk premia \(\pi_1\) and \(\pi_2\) as estimates of \(\hat{\pi}_1(\epsilon)\) and \(\hat{\pi}_2(\epsilon)\) at a small \(\epsilon\). Since utility functions considered here are strictly increasing in wealth (\(\frac{\partial V_i}{\partial W}>0\)), the definition in (\ref{def_pi}) implies that \(\hat{\pi}_i(0)=0\) and

\begin{equation}
    \hat{\pi}_i(\epsilon) \approx \hat{\pi}_i(0)+\frac{\partial \hat{\pi}_i}{\partial \epsilon} \epsilon = \pi_i \epsilon
\end{equation}
So, using \(\hat{\pi}_i(\epsilon)\) in equation (\ref{def_alphastar}) instead of \(\pi_i\) would be approximately equivalent.

Now, consider an alternative definition of  \(\hat{\pi}_i(\epsilon)\) as  the \textit{share} of the initial wealth agent \(i\) is willing to sacrifice in order to reduce the risky share by \(\epsilon\):
\begin{equation}
\label{def_pi_relative}
    V_i(W(1-\hat{\pi}_i(\epsilon)),\alpha-\epsilon)=V_i(W,\alpha)
\end{equation}
Equation (\ref{def_pi}) becomes:

\begin{equation}
\frac{\partial \hat{\pi}_i(\epsilon)}{\partial \epsilon} = -\frac{1}{W}\frac{\frac{\partial V_i}{\partial \alpha} }{\frac{\partial V_i}{\partial W}}    
\end{equation}
but the condition in (\ref{def_alphastar}) is unchanged as wealth cancels out. This means that, this alternative definition of the risk premium would result in the same \(\alpha^\star\). 

\subsubsection*{CARA Derivation}
Assume \(r_f=0\), \(\tilde{x} \sim \mathcal{N}(\mu, \sigma^2)\), and let \(v_i \equiv A_iln(-V_i)\). Then, using equation (\ref{def_pieps}), \(\hat{\pi}_i(\epsilon)\) can be equivalently defined by the following equation:

\begin{equation}
\label{def_pieps_2}
    v_i(W-\hat{\pi}_i(\epsilon),\alpha-\epsilon)=v_i(W,\alpha)
\end{equation}
where

\begin{equation}
\label{def_MV}
    v_i(W,\alpha) = -A_iW(1+\alpha\mu) +\frac{A_i^2W^2\alpha^2\sigma^2}{2}
\end{equation}
hence, we obtain:

\begin{equation}
    \pi_i = -\frac{\frac{\partial v_i}{\partial \alpha} }{\frac{\partial v_i}{\partial W}} = -\frac{-A_iW\mu+A_i^2W^2\alpha\sigma^2}{-A_i(1+\alpha\mu)+A_i^2W\alpha^2\sigma^2}=\frac{-W(\mu-\alpha A_iW\sigma^2)}{1+\alpha(\mu-\alpha A_iW\sigma^2)}
\end{equation}
Note that starting with Mean-Variance preferences would result in the same solution since we have a Mean-Variance equivalent objective in equation (\ref{def_MV})

\subsubsection*{CRRA Derivation}
\label{sec_AppERPCRRA}
Assume \(r_f=0\), then:
\begin{equation}
    V_i(W,\alpha) = \mathbb{E}\left[ \frac{(W(1+\alpha \tilde{x}))^{1-\gamma_i}}{1-\gamma_i}\right]
\end{equation}
hence, we obtain:
\begin{equation}
   \pi_i = -\frac{\frac{\partial V_i}{\partial \alpha} }{\frac{\partial V_i}{\partial W}}  = -\frac{W^{1-\gamma_i}\mathbb{E}[\tilde{x}(1+\alpha\tilde{x})^{-\gamma_i}]}{W^{-\gamma_i}\mathbb{E}[(1+\alpha\tilde{x})^{1-\gamma_i}]}=-W\frac{\mathbb{E}[\tilde{x}(1+\alpha\tilde{x})^{-\gamma_i}]}{\mathbb{E}[(1+\alpha\tilde{x})^{1-\gamma_i}]}
\end{equation}

\subsubsection*{Proposition \ref{prop_monotonicity_erp}}
\begin{proof}
    Assume \(r_f=0\). Let us start with the case of CARA. Note that equation (\ref{def_pi_cara}) implies that \(\pi_i\) always increases with \(A_i\) for a given \(\alpha\) 
\begin{equation*}
    \frac{\partial \pi_i}{\partial A_i} > 0
\end{equation*}
Also note that \(\pi_i\) always increases with \(\alpha\) for a given \(A_i\):
\[
\frac{\partial\pi_i }{\partial \alpha} = \frac{W\left[\left(\mu-\alpha A_iW\sigma^2\right)^2+A_iW\sigma^2\right]}{\left(1+\alpha\mu-\alpha^2A_iW\sigma^2\right)^2} >0
\]
Hence, using (\ref{compstat_erp}) we obtain
\[
\frac{\partial\alpha^\star}{\partial A_i}<0
\] 
Let us now consider the case of CRRA. First, we would like to show that \(\frac{\partial\pi_i }{\partial \alpha} > 0\). Using equation (\ref{def_pi_crra}), we get
\[
\begin{split}
\frac{\partial\pi_i }{\partial \alpha} = \\
W\frac{(1-\gamma_i)\mathbb{E}[\tilde{x}(1+\alpha\tilde{x})^{-\gamma_i}] \mathbb{E}[\tilde{x}(1+\alpha\tilde{x})^{-\gamma_i}]+\gamma_i\mathbb{E}[(1+\alpha\tilde{x})^{1-\gamma_i}]\mathbb{E}[\tilde{x}^2(1+\alpha\tilde{x})^{-\gamma_i-1}]}{\mathbb{E}[(1+\alpha\tilde{x})^{1-\gamma_i}]^2}
\end{split}
\]
Hence, the sign of \(\frac{\partial\pi_i }{\partial \alpha}\) is equal to the sign of the numerator which can be written as:
\[
-\gamma_i\left(\mathbb{E}[\tilde{x}(1+\alpha\tilde{x})^{-\gamma_i}]^2-\mathbb{E}[(1+\alpha\tilde{x})^{1-\gamma_i}]\mathbb{E}[\tilde{x}^2(1+\alpha\tilde{x})^{-\gamma_i-1}]\right)+\mathbb{E}[\tilde{x}(1+\alpha\tilde{x})^{-\gamma_i}]^2
\]
By Cauchy-Schwarz inequality
\[
\mathbb{E}[\tilde{x}(1+\alpha\tilde{x})^{-\gamma_i}]^2<\mathbb{E}[(1+\alpha\tilde{x})^{1-\gamma_i}]\mathbb{E}[\tilde{x}^2(1+\alpha\tilde{x})^{-\gamma_i-1}]
\]
Hence,
\[\frac{\partial\pi_i }{\partial \alpha} > 0\]
Next, we would like to show that \(\frac{\partial\pi_i }{\partial \gamma_i} > 0\):
First, let
\[
\tilde{\tilde{y}}\equiv 1+\alpha\tilde{x}>0
\]
Then \(\tilde{x}=\frac{\tilde{y}-1}{\alpha}\) and \(\pi_i\) can be written as:
\[
\pi_i = -\frac{W}{\alpha}\frac{\mathbb{E}[(\tilde{y}-1)\tilde{y}^{-\gamma_i}]}{\mathbb{\mathbb{E}}[\tilde{y}^{1-\gamma_i}]}=-\frac{W}{\alpha}\frac{\mathbb{E}[\tilde{y}^{1-\gamma_i}-\tilde{y}^{-\gamma_i}]}{\mathbb{\mathbb{E}}[\tilde{y}^{1-\gamma_i}]}=\frac{W}{\alpha}\left(\frac{\mathbb{E}[\tilde{y}^{-\gamma_i}]}{\mathbb{\mathbb{E}}[\tilde{y}^{1-\gamma_i}]}-1\right)
\]
Therefore, it suffices to show that \(\partial\frac{\mathbb{E}[\tilde{y}^{-\gamma_i}]}{\mathbb{\mathbb{E}}[\tilde{y}^{1-\gamma_i}]} / \partial \gamma_i >0\).
\[\partial\frac{\mathbb{E}[\tilde{y}^{-\gamma_i}]}{\mathbb{\mathbb{E}}[\tilde{y}^{1-\gamma_i}]} / \partial \gamma_i = 
\frac{-\mathbb{E}[\tilde{y}^{-\gamma_i}ln(\tilde{y})]\mathbb{E}[\tilde{y}^{1-\gamma_i}]+\mathbb{E}[\tilde{y}^{1-\gamma_i}ln(\tilde{y})]\mathbb{E}[\tilde{y}^{-\gamma_i}]} {(\mathbb{E}[\tilde{y}^{1-\gamma_i}])^2}\]
where the ratio is positive if the numerator is positive. Now define a probability distribution tilted by \(\tilde{y}^{-\gamma_i}\):
\[
dP_{\gamma_i}=\frac{\tilde{y}^{-\gamma_i}}{\mathbb{E}[\tilde{y}^{-\gamma_i}]}dP
\]
Under this distribution,
\[
\mathbb{E}_{\gamma_i}[g(\tilde{y})]=\frac{\mathbb{E}[\tilde{y}^{-\gamma_i}g(\tilde{y})]}{\mathbb{E}[\tilde{y}^{-\gamma_i}]}
\]
Hence, 
\[
\mathbb{E}[\tilde{y}^{-\gamma_i}g(\tilde{y})]=\mathbb{E}_{\gamma_i}[g(\tilde{y})]\mathbb{E}[\tilde{y}^{-\gamma_i}]
\]
and the numerator above can be written as,
\[
\begin{split}
\mathbb{E}[\tilde{y}^{1-\gamma_i}ln(\tilde{y})]\mathbb{E}[\tilde{y}^{-\gamma_i}]-\mathbb{E}[\tilde{y}^{-\gamma_i}ln(\tilde{y})]\mathbb{E}[\tilde{y}^{1-\gamma_i}]= \\
\mathbb{E}[\tilde{y}^{-\gamma_i}\tilde{y}ln(\tilde{y})]\mathbb{E}[\tilde{y}^{-\gamma_i}]-\mathbb{E}[\tilde{y}^{-\gamma_i}ln(\tilde{y})]\mathbb{E}[\tilde{y}^{-\gamma_i}\tilde{y}] = \\
\mathbb{E}_{\gamma_i}[\tilde{y}ln(\tilde{y})]\mathbb{E}[\tilde{y}^{-\gamma_i}]\mathbb{E}[\tilde{y}^{-\gamma_i}]-\mathbb{E}_{\gamma_i}[ln(\tilde{y})]\mathbb{E}[\tilde{y}^{-\gamma_i}]\mathbb{E}_{\gamma_i}[\tilde{y}]\mathbb{E}[\tilde{y}^{-\gamma_i}] = \\
(\mathbb{E}[\tilde{y}^{-\gamma_i}])^2(\mathbb{E}_{\gamma_i}[\tilde{y}ln(\tilde{y})]-\mathbb{E}_{\gamma_i}[\tilde{y}]\mathbb{E}_{\gamma_i}[ln(\tilde{y})])>0
\end{split}
\]
as \(\tilde{y}\) and \(ln(\tilde{y})\) are both increasing in \(\tilde{y}\). This means that the numerator is positive and \(\frac{\partial\pi_i }{\partial \gamma_i} > 0\).
Since we have also shown that \(\frac{\partial\pi_i }{\partial \alpha} > 0\), we obtain
\[
\frac{\partial\alpha^\star}{\partial \gamma_i}<0
\] using (\ref{compstat_erp}) \end{proof}

\subsubsection*{Proposition \ref{prop_bargainingerp}}

\begin{proof}
    Let \( F(\lambda_1,\alpha) \equiv\lambda_1\pi_1(\alpha)+(1-\lambda_1)\pi_2(\alpha) \). Then by the condition in (\ref{def_alphastar}), we have \(F(\lambda_1,\alpha^\star)=0\) and by the implicit function theorem:
    
    \begin{equation*}
            \frac{\partial \alpha^\star}{\partial \lambda_1} = -\frac{\frac{\partial F}{\partial\lambda_1}}{\frac{\partial F}{\partial \alpha}}
    \end{equation*}
Note that the denominator is positive:
\[\frac{\partial F}{\partial \alpha}=\lambda_1\frac{\partial\pi_1}{\partial\alpha}+(1-\lambda_1)\frac{\partial\pi_2}{\partial\alpha}>0\] 
and that the numerator is negative:
\[\frac{\partial F}{\partial \lambda_1}=\pi_1(\alpha^\star)-\pi_2(\alpha^\star)<0\] 
as \(\pi_1(\alpha^\star)<0\) and \(\pi_2(\alpha^\star)>0\). So, we obtain \(\frac{\partial \alpha^\star}{\partial \lambda_1}>0\)    
\end{proof}

\subsubsection*{Approximation}
Write \(\tilde{x}=k\hat{\mu}+\tilde{y}\) where \(\mathbb{E}[\tilde{y}]=0\). Assume \(k\) is positive but small (i.e. henceforth referred to as small risk assumption). Then using the Arrow-Pratt approach we get
\begin{equation}
\alpha^\star_i\approx \frac{\mu}{-W\frac{U_i''(W)}{U_i'(W)}\sigma^2}
\end{equation}
I am going to derive an approximation for \(\pi_i(\alpha)\) assuming small risk and show that
\begin{equation}
    \pi_i(\alpha) \approx -W^2\frac{U_i''(W)}{U_i'(W)}\sigma^2(\alpha-\alpha^\star_i)
\end{equation}
for a small \(\alpha-\alpha_i^\star\) or a small \(\alpha\).
First, Taylor expand \(\pi_i(\alpha)\) around \(\alpha_i^\star\), using \(\pi_i(\alpha_i^\star)=0\) :
\begin{equation}
    \pi_i(\alpha) \approx \pi_i(\alpha_i^\star) + \pi_{i,\alpha}(\alpha_i^\star)(\alpha-\alpha_i^\star) = \pi_{i,\alpha}(\alpha_i^\star)(\alpha-\alpha_i^\star)
\end{equation}
and using \(V_{i,\alpha}(\alpha_i^\star)=0\)
\begin{equation}
\pi_{i,\alpha}(\alpha_i^\star)=    -\frac{V_{i,\alpha \alpha}(\alpha_i^\star)V_{i,W}(\alpha_i^\star)-V_{i,\alpha}(\alpha_i^\star)V_{i,\alpha W}(\alpha_i^\star)}{V_{i,W}^2(\alpha_i^\star)}=-\frac{V_{i,\alpha \alpha}(\alpha_i^\star)}{V_{i,W}(\alpha_i^\star)}
\end{equation}
Now, using \( V_i(W,\alpha) = \mathbb{E} [U_i(W(1+\alpha\tilde{x}))]\), we have
\begin{equation*}
    \begin{split}
    V_{i,W} = \mathbb{E}[U_i'(W(1+\alpha\tilde{x}))(1+\alpha\tilde{x})] \\
    V_{i,\alpha} = W\mathbb{E}[U_i'(W(1+\alpha\tilde{x}))\tilde{x}] \\
    V_{i,\alpha W} = W\mathbb{E}[U_i''(W(1+\alpha\tilde{x}))\tilde{x}(1+\alpha\tilde{x})]+\mathbb{E}[U_i'(W(1+\alpha\tilde{x}))\tilde{x}] \\
    V_{i,\alpha \alpha} = W^2\mathbb{E}[U_i''(W(1+\alpha\tilde{x}))\tilde{x}^2]
    \end{split}
\end{equation*}
So, because \(k\) is small, we have \(\alpha_i^\star\tilde{x}\approx0\) and:
\begin{equation*}
    \begin{split}
    V_{i,W}(\alpha_i^\star) \approx U_i'(W) \\
    V_{i,\alpha \alpha}(\alpha_i^\star) \approx W^2U_i''(W)\mathbb{E}[\tilde{x}^2] \approx W^2U_i''(W)\sigma^2
    \end{split}
\end{equation*}
hence we obtain,
\begin{equation*}
    \pi_i(\alpha) \approx -W^2\frac{U_i''(W)}{U_i'(W)}\sigma^2(\alpha-\alpha^\star_i)
\end{equation*}
Now, Taylor expand \(\pi_i(\alpha)\) around \(0\)
\begin{equation*}
    \pi_i(\alpha) \approx \pi_i(0) + \pi_{i,\alpha}(0)\alpha
\end{equation*}
where
\begin{equation*}
\pi_i(0)=    -\frac{V_{i,\alpha}(0)}{V_{i,W}(0)}=-\frac{W\mathbb{E}[U_i'(W)\tilde{x}]}{\mathbb{E}[U_i'(W)]}=-W\mu
\end{equation*}
and
\begin{equation*}
\begin{split}
\pi_{i,\alpha}(0)=    -\frac{V_{i,\alpha \alpha}(0)V_{i,W}(0)-V_{i,\alpha}(0)V_{i,\alpha W}(0)}{V_{i,W}^2(0)} =\\
-\frac{W^2U_i''(W)\mathbb{E}[\tilde{x}^2]U_i'(W)-W\mathbb{E}[\tilde{x}]U_i'(W)[WU_i''(W)\mathbb{E}[\tilde{x}]+U_i'(W)\mathbb{E}[\tilde{x}]]}{U_i'(W)^2}
\end{split}
\end{equation*}
Note that under the small risk assumption we have \(\mathbb{E}[\tilde{x}]^2\approx0\), \(\mathbb{E}[\tilde{x}^2]\approx\sigma^2\), and
\begin{equation}
\begin{split}
\pi_{i,\alpha}(0)=-W^2\frac{U_i''(W)}{U_i'(W)}\sigma^2
\end{split}
\end{equation}
hence, we again obtain
\begin{equation}
    \pi_i(\alpha) \approx -W\mu -W^2\frac{U_i''(W)}{U_i'(W)}\sigma^2\alpha \approx  -W^2\frac{U_i''(W)}{U_i'(W)}\sigma^2(\alpha-\alpha^\star_i)
\end{equation}

\subsection{Wealth Splitting Problem}
\label{sec_AppWealthSplit}
Consider the following simple problem in which there is no uncertainty and the household decides how to split some wealth \(W>0\) between the two members in order to maximize the weighted sum of utilities. Each member is assumed to only consume some private good \(C_1+C_2=W\), so the household problem is

\[\max_{\omega} \;\lambda_1U_1(\omega W)+\frac{\lambda_2}{\delta}U_2((1-\omega) W)\] 

Here, \(\delta\) is a normalization parameter. Observe that under a constant \(\delta\) the risk aversion of a member affects the optimal split even though there is no uncertainty in this framework. At a sufficiently high (low) wealth level, the share allocated to person \(i\) decreases (increases) with her risk aversion \(\gamma_i\). I will define \(\delta\) as a function of \(W\) and the risk aversion parameters such that the split \(\omega\) is independent of the difference in risk aversion parameters. More precisely the split that would be achieved in the utilitarian framework if preferences equal the average of the two actual ones. \cite{mazzocco2004saving}, defines delta in a similar spirit: such that the household splits resources evenly for a given level of wealth. Let \(C_1=\omega W\) and \(C_2=(1-\omega) W\).

\subsubsection*{CARA}

\begin{equation}
    \max_{\omega} \; \lambda_1 \frac{-exp(-A_1\omega W)}{A_1} + \frac{\lambda_2}{\delta} \frac{-exp(-A_2(1-\omega)W)}{A_2}
    \label{obj_splitting_sum_CARA}
\end{equation}
The F.O.C. is:
\begin{equation*}
    \lambda_1 exp(-A_1\omega W)=\frac{\lambda_2}{\delta}exp(-A_2(1-\omega)W)
\end{equation*}
taking logarithms:
\begin{equation*}
    ln \lambda_1 -A_1\omega W = ln\left(\frac{\lambda_2}{\delta}\right)-A_2(1-\omega)W
\end{equation*}
solving for the optimal split we obtain:
\begin{equation}
    \omega^\star = \frac{A_2}{A_1+A_2}+\frac{1}{(A_1+A_2)W}ln\left( \frac{\delta\lambda_1}{\lambda_2}\right)
\end{equation}
So, just as in the portfolio choice problem under unnormalized utilitarian approach, the difference in risk aversions affects the optimal solution and making person 2 more risk averse decreases her share.

Note that if \(\delta=1\) and both individuals have identical risk aversion \(A=\frac{A_1+A_2}{2}\), the optimal split becomes:

\begin{equation}
    \bar{\omega} = \frac{1}{2}+\frac{1}{(A_1+A_2)W}ln\left( \frac{\lambda_1}{\lambda_2}\right)
\end{equation}
Here, the equal split is achieved when \(\lambda_1=\lambda_2\) irrespective of the risk aversion parameters. We can now ask what \(\delta\) is required such that \(\omega^\star=\bar{\omega}\):
\begin{equation}
    \delta = exp\left(\frac{W}{2}(A_1-A_2)\right) = \frac{exp(-A_2W/2)}{exp(-A_1W/2)}
\end{equation}
Note, that this normalization essentially divides the Pareto weights by the marginal utilities of wealth of the corresponding spouse at the equal split. 

Note that this also helps understand the pattern discussed above as when \(W=0\), the \(\delta\) becomes independent of the risk aversion parameters and is equal to \(1\). 

Redefine the setup slightly such that \( C_h=\alpha\tilde{x}\) instead of \(C_h=W(1 +\alpha \tilde{x})\). That is, initial wealth can be thought of as being set to 0 and shocks affecting the wealth additively instead of multiplicatively. Then, under equal Pareto weights and the utilitarian approach we get a solution to the household risky share identical to the one in equation (\ref{sol_harmonic_CARA}) under equal weights:
\begin{equation}
    \alpha^\star = \frac{\mu}{\frac{(A_1+A_2)}{2}W\sigma^2} = \frac{1}{\frac{1}{2\alpha_1^\star}+\frac{1}{2\alpha_2^\star}}
\end{equation}
Note that at \(C_h=0\), \(U_i'=exp(0)=1\). That is, similar to the CRRA utility at \(C=1\), the marginal CARA utility at \(C=0\) does not depend on the risk aversion parameter.

\subsubsection*{CRRA}
\begin{equation}
    \max_{\omega} \; \lambda_1 \frac{(\omega W)^{1-\gamma_1}}{1-\gamma_1} + \frac{\lambda_2}{\delta} \frac{((1-\omega)W)^{1-\gamma_2}}{1-\gamma_2}
\end{equation}
The F.O.C. is:
\begin{equation}
    \lambda_1(\omega W)^{-\gamma_1} = \frac{\lambda_2}{\delta}((1-\omega) W)^{-\gamma_2}
\end{equation}
Note that if \(\delta=1\) and both individuals have identical risk aversion \(\gamma=\frac{\gamma_1+\gamma_2}{2}\), the optimal split becomes:
\begin{equation}
\left(\frac{\bar{\omega}}{1-\bar{\omega} }\right)^\frac{\gamma_1+\gamma_2}{2} = \frac{\lambda_1}{\lambda_2}
\end{equation}
Here, the equal split is achieved when \(\lambda_1=\lambda_2\) irrespective of the risk aversion parameters. We can now ask what \(\delta\) is required such that \(\omega^\star=\bar{\omega}\):
\begin{equation}
\delta = \left[ \frac{W\left(\dfrac{\lambda_1}{\lambda_2}\right)^{
\dfrac{1}{\gamma_1+\gamma_2}}
} {1+ \left(\dfrac{\lambda_1}{\lambda_2}\right)^{\dfrac{2}{\gamma_1+\gamma_2}}}
\right]^{\gamma_1-\gamma_2} = W^{\gamma_1-\gamma_2}
\left(\frac{\lambda_1}{\lambda_2}\right)^{
\frac{\gamma_1-\gamma_2}{\gamma_1+\gamma_2}
}
\left[
1+
\left(\frac{\lambda_1}{\lambda_2}\right)^{
\frac{2}{\gamma_1+\gamma_2}
}
\right]^{\gamma_2-\gamma_1}
\end{equation}

Note that when \(W=2\) and \(\lambda_1=\lambda_2\) the \(\delta\) becomes independent of the risk aversion parameters and is equal to \(1\). This is consistent with Proposition \ref{prop_monotonicity_sum_CRRA} because here each member is consuming one unit of the good. Using this \(\delta\) in the utilitarian approach (with \(2W\) instead of \(W\)) in the portfolio choice problem yields approximately the same solution as in equation (\ref{sol_harmonic_CARA}).

\subsection{Expected Product of Utilities}
\label{sec_AppExpProd}
First, assume the household maximizes the expected product of CARA utilities. The objective can then be written as

\begin{equation}
    \mathbb{E} \left[ \prod_{i} \left(   \frac{ - exp(-A_i C_h)}{A_i} \right)^{\lambda_i} \right] 
         = \frac{\mathbb{E} \left[-exp(-\sum_i\lambda_i A_iC_h) \right]}{A_1^{\lambda_1} A_2^{\lambda_2}}
\end{equation}
This is equivalent to the individual CARA maximization problem with a renamed coefficient of risk aversion. Let \( A_{h} = \sum_i\lambda_i A_i\). Then, under normally distributed excess return, the optimal risky share is:
\[\alpha^\star_i = \frac{\mu}{A_hW\sigma^2}\]

The same arguments can be used in the case of CRRA utilities. Assume the household maximizes the expected product of CRRA utilities. The objective can then be written as

\begin{equation}
    \mathbb{E} \left[ \prod_{i} \left(   \frac{C_h^{1-\gamma_i}}{1-\gamma_i} \right)^{\lambda_i} \right] 
         = \frac{\mathbb{E} \left[C_h^{\lambda_1(1-\gamma_1)+\lambda_2(1-\gamma_2)} \right]}{(1-\gamma_1)^{\lambda_1}(1-\gamma_2)^{\lambda_2}}=  \frac{\mathbb{E} \left[C_h^{1-\sum_i \lambda_i\gamma_i} \right]}{(1-\gamma_1)^{\lambda_1}(1-\gamma_2)^{\lambda_2}}
\end{equation}
This is equivalent to the individual CRRA maximization problem with a renamed coefficient of risk aversion. Let \( \gamma_{h} = \sum_i\lambda_i \gamma_i\). Then under the small risk assumption, the optimal risky share is:
\[\alpha^\star \approx \frac{\mu}{\gamma_h\sigma^2}\]

\subsection{Other criteria}
\label{sec_AppOtherApp}
In this section, I will consider aggregation of individual utilities by the product of their expectations: \(\mathbb{E}[U_h]=\prod_i\mathbb{E}[U_i]^{\lambda_i}\), as well as the Nash Bargaining approach. In this subsection, I focus mostly on CARA individual utilities. Assume the household maximizes a weighted product of expected CARA utilities:
\begin{equation}
    \max_{\alpha} \; \mathbb{E} [U_{h}(C_h)] = \max_{\alpha} \; \prod_{i} \mathbb{E}[U^{cara}_i(C_h)]^{\lambda_i} 
    \label{obj_ProdExp_CARA}
\end{equation}

\begin{proposition}
\label{prop_prodexp}
        Let \(\theta = (\mu,\sigma,r_f,W,\lambda_1,\lambda_2,A_j) \). Then, there exists \(\theta  \in \mathbb{R}_{+}^7 \) such that solution \(\alpha^\star\) to the household problem in (\ref{obj_ProdExp_CARA}) is a non-monotonic function of \(A_i\).
\end{proposition}

\begin{proof} The objective can be written as:
\begin{equation}
\begin{split}
  \mathbb{E} [U_{h}(C_h)] = 
    \prod_{i} \left( \mathbb{E} \left[  \frac{ - exp(-A_i C_h)}{A_i} \right] \right)^{\lambda_i}   =   \prod_{i} \left( \frac{ \mathbb{E} \left[   - exp(-A_i C_h)\right] }{A_i} \right)^{\lambda_i}
\end{split}
\end{equation}

Applying the property of exponents we get:

\begin{equation}
\begin{split}
    \mathbb{E} (U_{h}) =  \prod_{i}  \left( \frac{-exp \left(-A_iW(1+r_f+\alpha\mu) +\frac{A_i^2W^2\alpha^2\sigma^2}{2} \right) }{A_i} \right)^{\lambda_i} \\
    =\frac{exp \left(-(\lambda_1A_1+\lambda_2A_2)W(1+r_f+\alpha\mu) +\frac{(\lambda_1A_1^2+\lambda_2A_2^2)W^2\alpha^2\sigma^2}{2} \right) }{A_1^{\lambda_1}A_2^{\lambda_2}}
\end{split}
    \label{obj_ProdExp_CARA_2}
\end{equation}

The F.O.C. of (\ref{obj_ProdExp_CARA_2}) is (the denominator can be ignored as it does not depend on \(\alpha\)):

\begin{equation}
    \begin{split}
        exp \left(-(\lambda_1A_1+\lambda_2A_2)W(1+r_f+\alpha\mu) +\frac{(\lambda_1A_1^2+\lambda_2A_2^2)W^2\alpha^2\sigma^2}{2} \right) \times \\
        \left( -(\lambda_1A_1+\lambda_2A_2)W\mu + (\lambda_1A_1^2+\lambda_2A_2^2)W^2\alpha\sigma^2 \right) = 0
    \end{split}
\end{equation}

Since the first term in the multiplication above is always positive, the second term has to be equal to 0. Hence, we obtain:
\begin{equation}
\label{sol_ProdExp_CARA}
 \alpha^\star = \frac{(\lambda_1A_1+\lambda_2A_2)\mu}{(\lambda_1A_1^2+\lambda_2A_2^2)W\sigma^2    }
\end{equation}
The solution is decreasing in \(A_2\) if and only if:

\begin{equation}
\label{cond_ProdExp_CARA}
A_2 > A_1 \frac{\sqrt{\lambda_1(\lambda_1+\lambda_2)} - \lambda_1}{\lambda_2}    
\end{equation}
This inequality does not hold for all \(A_1,A_2,\lambda_1,\lambda_2\). \end{proof}

Assuming \(\lambda_1+\lambda_2=1\), the condition for monotonicity with respect to \(A_2\) can be written as, \( A_2 > A_1 \frac{\sqrt{\lambda_1}(1 - \sqrt{\lambda_1})}{1-\lambda_1} \). Since \(\frac{\sqrt{\lambda_1}(1 - \sqrt{\lambda_1})}{1-\lambda_1}<0.5 \ \ \forall \lambda_1\), monotonicity is guaranteed when \(A_2>0.5A_1\). Note that this is in contrast to the result under sum of expected utilities where we can guarantee monotonicity with respect to \(A_2\) when \(A_2<A_1\).

Motivated by \cite{gu2024gender}, I also consider the Mean-Variance utility which is a monotonic transformation of the CARA utility:
\[ U^{mv}_i(C_i) = A_i\mathbb{E}[C_i]-\frac{1}{2}A_i^2 Var(C_i)\]
Assume household maximizes weighted sum of individual Mean-Variance utilities. That is the household problem is:

 \begin{equation}
      \max_{\alpha} \; \mathbb{E} [U_{h}(C_h)] =  \sum_i \lambda_i \mathbb{E} (U^{mv}_i)  = \sum_i\lambda_i A_iW( 1+r_f+\alpha \mu) - \frac{1}{2} \sum_i\lambda_i A_i^2 W^2\alpha^2 \sigma^2
 \label{obj_sum_MV}
 \end{equation}

\begin{proposition}
\label{prop_summv}
    Maximizing weighted sum of Mean-Variance utilities is equivalent to maximizing weighted product of expected CARA utilities.
\end{proposition}

\begin{proof}
The F.O.C. of equation (\ref{obj_sum_MV}) is:
\begin{equation}
    \sum_i\lambda_i A_i \mu - \sum_i\lambda_i A_i^2W \alpha^\star \sigma^2 = 0
\end{equation}
So, we obtain\footnote{This means that the implicit household-level risk aversion is \( A_{h} = \frac{\sum_i\lambda_i A_i^2}{\sum_i\lambda_i A_i} \) }
\begin{equation}
 \alpha^\star = \frac{(\lambda_1A_1+\lambda_2A_2)\mu}{(\lambda_1A_1^2+\lambda_2A_2^2)W\sigma^2    }
\end{equation}
This is identical to the product of expected CARA utilities solution given in equation (\ref{sol_ProdExp_CARA}).
\end{proof}

This equivalence can be observed by taking the logarithm of the objective function in (\ref{obj_ProdExp_CARA}) and comparing it to the objective in (\ref{obj_sum_MV}). Proposition \ref{prop_summv} demonstrates that while under normally distributed risky return maximizing Mean-Variance utility is equivalent to maximizing expected CARA utility under the unitary framework, this is not the case under the collective framework. Now, redefine the objective in (\ref{obj_sum_MV}) slightly such that the household maximizes weighted sum of individual simplified Mean-Variance utilities (as in Gu et al., 2024) where the risk aversion parameter has been canceled out:
\[ U^{smv}_i(C_i) = \mathbb{E}[C_i]-\frac{1}{2}A_i Var(C_i) \]
The household problem is:

 \begin{equation}
      \max_{\alpha} \; \mathbb{E} [U_{h}(C_h)] =  \sum_i \lambda_i \mathbb{E} (U^{smv}_i)  = \sum_i\lambda_iW ( 1+r_f+\alpha \mu) - \frac{1}{2} \sum_i\lambda_i A_i W^2\alpha^2 \sigma^2
 \label{obj_sum_sMV}
 \end{equation}

\begin{proposition}
\label{prop_sumsmv}
    Maximizing weighted sum of simplified Mean-Variance utilities is equivalent to maximizing expectation of a weighted product of CARA utilities.
\end{proposition}
\begin{proof}
The F.O.C. of equation (\ref{obj_sum_sMV}) is:
\begin{equation}
    \sum_i\lambda_i \mu - \sum_i\lambda_i A_i W\alpha^\star \sigma^2 = 0
\end{equation}

Assuming \(\sum_i\lambda_i=1\),
\begin{equation}
 \alpha^\star = \frac{\mu}{\sum_i\lambda_i A_i W \sigma^2}
\end{equation}
This is identical to the expected product of CARA utilities solution given in Section \ref{subsection_expprod}. \end{proof}

Proposition \ref{prop_sumsmv} demonstrates that while under  the unitary framework maximizing Mean-Variance utility is equivalent to maximizing the simplified Mean-Variance utility, it is not the case under the collective framework. 

Finally, consider the Nash Bargaining problem:

\begin{equation}
    \max_{\alpha} \;\mathbb{E} [U_{h}(C_h)] = \max_{\alpha} \; \prod_{i} \left( \mathbb{E}[U^{cara}_i(C_h)] - \Lambda_i\right)^{\lambda_i}
    \label{obj_NB_CARA}
\end{equation}
The following proposition shows that a Nash Bargaining solution can also be a non-monotonic function of the risk aversion of a member.

\begin{proposition}
\label{prop_nb}
       Let \(\theta = (\mu,\sigma,r_f,W,\lambda_1,\lambda_2,A_j) \). Then, there exists \(\theta  \in \mathbb{R}_{+}^7 \)  and \(\Lambda_1, \Lambda_2 \neq 0\) such that solution \(\alpha^\star\) to the household problem in (\ref{obj_NB_CARA}) is a non-monotonic function of \(A_i\).
\end{proposition}
\begin{proof}

\begin{figure}[th]
\centering

    \includegraphics[width=0.45\linewidth]{"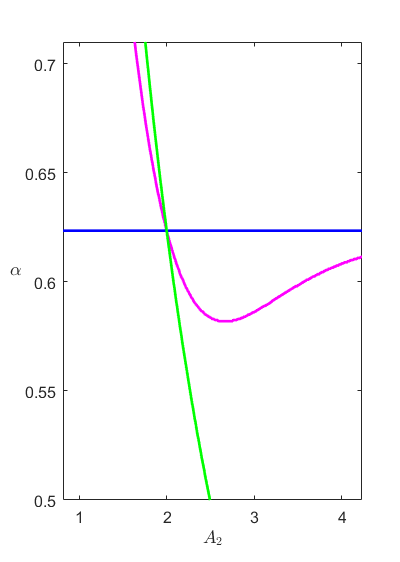"}
    \caption{\footnotesize 
    This figure plots the optimal risky share \(\alpha\), against the risk aversion of the second member of the household for CARA individual utility functions. The blue curve plots \(\alpha^\star_1\) and the green curve plots \(\alpha^\star_2\) which are the solutions to the single's problem and are defined in equation (\ref{obj_unitary}). The magenta curve depicts the household-level optimal risky share \(\alpha^\star\), which solves the objective in (\ref{obj_nb}). Set, \(\mu= 0.1, \ \sigma=0.2, \ r_f=0, \ W=2, \ \lambda_1=\lambda_2=0.5, \ \Lambda_1=\Lambda_2=-1\), and \(A_1=2\).}
    \label{fig_nonmonotonicity_NB}
\end{figure}

The objective is:
\begin{equation}
\label{obj_nb}
   \mathbb{E} [U_{h}(C_h)] = \prod_i \left( \frac{-exp \left(-A_iW(1+r_f+\alpha\mu) +\frac{A_i^2W^2\alpha^2\sigma^2}{2} \right)}{A_i} -\Lambda_i \right)^{\lambda_i} 
\end{equation}

Despite the fact that the F.O.C. is analytical, there is no closed form solution for \(\alpha^\star\). Figure \ref{fig_nonmonotonicity_NB} visualizes a solution to an example which proves the proposition.
\end{proof}

\end{document}